\documentclass[prx,twocolumn,aps,superscriptaddress,floatfix,nofootinbib,a4paper,10pt]{revtex4-2}
\usepackage{graphicx} 
\usepackage[dvipsnames]{xcolor}
\usepackage{hyperref}
\hypersetup{
    colorlinks = true,
    allcolors = RoyalPurple
}
\usepackage{amsfonts,amsmath,amssymb,amsthm}
\usepackage{marvosym}
\usepackage{dsfont}
\usepackage{braket}
\usepackage{tikz}
\usepackage{complexity}
\usepackage{orcidlink}

\DeclareMathOperator{\sgn}{sgn}
\DeclareMathOperator{\Mat}{Mat}

\newtheorem{definition}{Definition}[]
\newtheorem{problem}{Problem}
\newtheorem{lemma}[definition]{Lemma}
\newtheorem{theorem}{Theorem}

\newcommand{\comm}[2]{\left[ #1 , #2 \right]}

\renewcommand{\braket}[1]{\left\langle #1 \right\rangle}
\renewcommand{\ket}[1]{\left\vert #1 \right\rangle}
\renewcommand{\bra}[1]{\left\langle #1 \right\vert}

\newcommand{\norm}[1]{\left\lVert #1 \right\rVert}
\newcommand{\abs}[1]{\left\lvert #1 \right\rvert}

\renewcommand{\epsilon}{\varepsilon}

\definecolor{BLUE}{RGB}{0,0,255}

\newcommand{\QOverview}{
\begin{tikzpicture}
    \filldraw[fill=WildStrawberry!20, align=center] (0,0) ellipse (3.5cm and 3cm) node {\tiny Gaussian Hamiltonians \\[4.5cm]};
    \def\ICellipse{(0,0.025) ellipse (2.5cm and 2cm)};
    \filldraw[fill=RoyalBlue!40, align=center] \ICellipse node {\\[2ex] Inertially Coupled Bosons};
    \def\Hoppingellipse{(0,-1.95) ellipse (1.7cm and 1cm)};
    \filldraw[fill=RoyalBlue!30, align=center] \Hoppingellipse node {\\[0.5cm] \tiny Hopping \\[-1ex] \tiny Hamiltonians};
    \begin{scope}
        \clip \ICellipse;
        \fill[RoyalBlue!20, align=center] \Hoppingellipse;
    \end{scope}
    \node[align=center] at (0,-1.5) {\tiny Quantum \\[-1ex] \tiny Walks};
    \draw[dashed] (1.54,-1.56) arc (-52:-128:2.5cm and 2cm);
    \filldraw[fill=RoyalBlue!30, align=center] (0,1.2) ellipse (1.6cm and 0.8cm) node {\\[0.65cm] \tiny Quantum Oscillators};
    \filldraw[fill=RoyalBlue!20, align=center] (0,1.5) ellipse (1cm and 0.5cm) node {\tiny Classical \\[-1ex] \tiny Oscillators};
    \node at (2,0.1) {\tiny $E=0$};
    \node at (1,1) {\tiny $C=\mathds{1}$};
    \node at (1,-1.5) {\tiny $A=C$};
    \node at (0,-2.75) {\tiny $E^T=-E$};
\end{tikzpicture}
}

\date{June 10, 2026}

\begin{abstract}
The computational complexity of simulating the dynamics of physical quantum
systems is a central question at the interface of quantum physics and computer
science. In this work, we address this question for the simulation of
exponentially large bosonic Hamiltonians with quadratic interactions. We present
two results: First, we introduce a broad class of quadratic bosonic problems
for which we prove that they are \BQP-complete. Importantly, this class includes
two known \BQP-complete problems as special cases: Classical oscillator
networks and continuous-time quantum walks. Second, we show that extending the
aforementioned class to even more general quadratic Hamiltonians results in a
\PostBQP-complete problem, which constitutes a jump in the complexity of
simulating large quantum systems on a quantum computer, as well as in the
difference in complexity between their simulation on classical and quantum computers.
\end{abstract}

\begin{document}
\title{Complexity of Quadratic Bosonic Hamiltonian Simulation: \\ \BQP-Completeness and \PostBQP-Hardness}

\author{Lilith Zschetzsche\,\orcidlink{0009-0002-9531-4039}}
\affiliation{University of Vienna, Faculty of Physics, Boltzmanngasse 5, 1090 Vienna, Austria}

\author{Refik Mansuro\u{g}lu\,\orcidlink{0000-0001-7352-513X}}
\email[]{Refik.Mansuroglu@univie.ac.at}
\affiliation{University of Vienna, Faculty of Physics, Boltzmanngasse 5, 1090 Vienna, Austria}

\author{Norbert Schuch\,\orcidlink{0000-0001-6494-8616}\,}
\affiliation{University of Vienna, Faculty of Physics, Boltzmanngasse 5, 1090 Vienna, Austria}
\affiliation{\mbox{University of Vienna, Faculty of Mathematics, Oskar-Morgenstern-Platz 1, 1090 Vienna, Austria}}

\maketitle

Quadratic bosonic (i.e., Gaussian) Hamiltonians  sit between quantum and
classical systems. On the one hand, they exhibit clear quantum features, such as
sub-classical noise due to squeezing. On the other hand, the resulting
Heisenberg equations are closed on the level of single bosonic operators, and
thus only couple moments of the same order. Therefore, given moments of some
order $R$---which can be encoded in a vector of dimension $d=O(M^R)$---the
dynamics of a system of $M$ bosonic modes can be simulated efficiently (i.e., with
resources $\poly(M)$). In particular, this implies that these systems cannot
capture the full power of quantum computing (unless \BQP=\BPP, or moments of
unbounded order are involved).

That the dynamics of these systems is captured by a linear first-order
differential equation of size $d=O(M^R)$ raises an intriguing question: Can we
gain an advantage by simulating these restricted quantum systems on a fully
fledged quantum computer, given that the problem can be represented using $\log
d=O(R \log M)$ qubits? For the related problem of quadratic \emph{fermionic}
Hamiltonians, this has indeed been demonstrated~\cite{Kraus2011}. 
%
%
As for bosons, Babbush \emph{et al.}~\cite{Babbush_2023} have recently shown how to
simulate exponentially large \emph{classical} systems with quadratic Hamiltonians
on a quantum computer, and that this problem is, in fact, \BQP-complete. As the evolution
equations for the first moments are identical for the classical and quantum
Hamiltonian, this implies \BQP-completeness also for the
corresponding class of bosonic quantum Hamiltonians. A seemingly unrelated
problem which allows for exponential compression and is \BQP-complete are
continuous-time quantum walks, which have long been used as a way
to reason about quantum computing~\cite{Childs_2003,Childs_2009}.  Remarkably,
quantum walks can also be understood as a special instance of bosonic Gaussian
Hamiltonians---namely, a quadratic hopping Hamiltonian acting on the subspace
with one boson.

That for these two simple and distinct classes of Gaussian Hamiltonians,
exponential-size instances are \BQP-complete, raises questions about the
computational complexity of general bosonic Gaussian Hamiltonians: What is the
most general class of bosonic Gaussian Hamiltonians for which exponentially
large instances are \BQP-complete---ideally including both of the aforementioned
problems as special cases? And are there limits to this---are all bosonic
Gaussian Hamiltonians exponentially compressible, or are there classes of
quadratic bosonic Hamiltonians which---despite the fact that they are
classically efficiently simulatable---no longer allow for an exponential
advantage when simulated on a quantum computer, and for which the exponential
gap between classical and quantum simulation therefore breaks down?

In this work, we introduce a class of quadratic bosonic Hamiltonians without
position-momentum coupling, which we term \textit{inertially coupled bosons},
and prove that simulating their dynamics of  exponentially many modes is
\BQP-complete. This class on the one hand contains systems of quantum harmonic
oscillators and thus subsumes the classical
oscillator networks studied in Ref.~\cite{Babbush_2023}. At the same time, we
show that the Hamiltonians governing continuous-time quantum walks can be
realized as a particular form of inertial coupling, evolving a state with a
single bosonic excitation. Inertially coupled bosons thus provide a single
framework which unifies and generalizes classical oscillator dynamics and quantum walk computation within
the same complexity theoretic setting. 
We further extend the class of inertially coupled bosons to all particle
number–conserving quadratic (``hopping'') Hamiltonians\footnote{In contrast to
quantum walks and inertially coupled bosons, these Hamiltonians may include
position–momentum couplings.} while retaining \BQP-completeness. 

Subsequently, we proceed to address the question whether all quadratic
Hamiltonians can be placed inside \BQP, and answer it in the negative: We show
that allowing additional quadratic terms beyond inertial couplings and hopping
interactions (the classes discussed above) leads to a sudden jump in the
computational complexity of the problem. Specifically, the inclusion of such
additional terms makes the problem  \PostBQP-complete (equivalently,
\PP-complete), which, remarkably, is identical to the complexity of classically
simulating the very same system. Whether every intermediate extension beyond the \BQP-complete subclasses already exhibits \PostBQP-hardness, or whether there exist quadratic bosonic Hamiltonian families of intermediate complexity, remains open.

\begin{figure}[t]
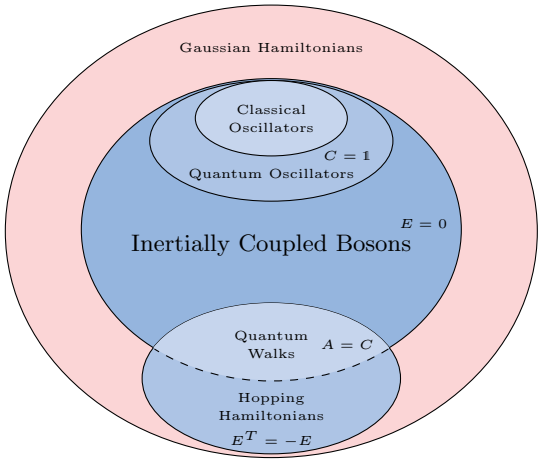

    \centering
    \QOverview
   \caption{\textbf{Venn diagram depicting inclusion of different classes of Hamiltonians in the class of quadratic bosonic Hamiltonians, cf.\ Eq.~\eqref{eq:QO}.} We define the class of inertially coupled bosons as a subset of quadratic bosonic Hamiltonians and prove that their simulation can be formulated as a \BQP-complete decision problem. This unifies previous proofs for \BQP-completeness of classical and quantum oscillators as well as quantum walks and general hopping Hamiltonians extending the class of inertially coupled bosons. Shades of blue indicate \BQP-completeness and Hamiltonians in red define a \PostBQP-complete problem, in general.}
    \label{fig:QO}
\end{figure}
\emph{Inertially Coupled Bosons.---}We consider a system of $M = 2^m$ bosonic modes with position and momentum operators $q_j$ and $p_j$ (with mode label $j$) satisfying the bosonic commutation relations
$\comm{q_j}{p_k} = i \delta_{jk} \mathds{1}$, 
$\comm{q_j}{q_k}=\comm{p_j}{p_k}=0$.
A system of inertially coupled bosons is defined by the quadratic Hamiltonian 
\begin{align}
    H = \frac{1}{2} \sum_{j,k} (C_{jk} p_j p_k + A_{jk} q_j q_k),
    \label{eq:HOHam}
\end{align}
where $A, C \in \Mat(M \times M, \mathds{R})$ are symmetric $d$-sparse matrices,
i.e.\ each row and column contains at most $d$ nonzero entries. For now, we
require $A$ and $C$ to be positive semidefinite (PSD), amounting to a lower
semibounded spectrum of the bosonic Hamiltonian. Further details, including
extensions beyond the PSD case, are provided in Appendix \ref{app:PSD}.

For our construction, we impose a sign convention on the elements of $A$, 
motivated by a potential energy arising from spring-like coupling, i.e., of the form
\begin{align}
    \frac{1}{2} \sum_{j,k} k_{jk} (q_j - q_k)^2 + \frac{1}{2} \sum_j k_{jj} q_j^2
    \label{eq:interaction_terms}
\end{align}
with positive coupling constants $k_{jk} = - A_{jk}$ for interactions between the $j^\text{th}$ and $k^\text{th}$ boson, and positive coefficients 
\begin{align}
    k_{jj} = A_{jj} + \sum_{j' \neq j} A_{jj'}
\end{align}
for the coupling to a reference point,\footnote{Nonzero $k_{jj}$ restrict the motion of the bosons and prevent arbitrary displacements from the reference point. This can model weak boundaries of an effective confining volume.}
and analogously for $C$ as arising from a kinetic energy of the form
\begin{align}
    \frac{1}{2} \sum_{j,k} l_{jk} (p_j - p_k)^2 + \frac{1}{2} \sum_j l_{jj} p_j^2,
    \label{eq:interaction_terms_2}
\end{align}
with coupling constants $l_{jk} \geq 0$ for all $j,k$. This is equivalent to requiring that $A$ and $C$ are Laplacian-type stiffness matrices, which are positive semidefinite and diagonally dominant.

With this sign convention, the model corresponds to the quantum analog of a classical Hamiltonian derived from couplings via positive spring constants, now acting in both position and momentum space. Special cases include a quantum analog of a spring network when $C = M^{-1}$ is the inverse of the (diagonal) mass matrix $M$, and (possibly multiparticle) quantum walks when $A = C$, which we discuss in detail later. 

Off-diagonal entries of $C$, corresponding to coupling between the momenta of different bosons, can be interpreted as inertial coupling. Such couplings naturally arise in constrained classical systems, for example in the classical double pendulum.

Note that the Hamiltonian in Eq.~\eqref{eq:HOHam} is in general not particle-number conserving. Yet, it remains 2-local (i.e.\ each term has support on at most two modes). We will later show that the \BQP-complete class can be extended to include general particle-number–preserving quadratic Hamiltonians, including those with complex hopping amplitudes. For clarity of presentation, we focus here on the inertially coupled case corresponding to Eq.~\eqref{eq:HOHam}.

Going beyond inertially coupled bosons, the most general quadratic bosonic
Hamiltonian, 
    \begin{align}
        H = \frac{1}{2} \begin{pmatrix} q & p \end{pmatrix} \begin{pmatrix} A & E \\ E^T & C \end{pmatrix} \begin{pmatrix} q \\ p \end{pmatrix}\ ,
        \label{eq:QO}
    \end{align}
includes mixed position–momentum couplings described by a matrix $E_{jk} =
E^+_{jk}+iE^-_{jk}$,  which can be split into a real symmetric component $E^+$
and a real antisymmetric component $iE^-$. These mixed $q$-$p$ terms include both
general hopping interactions and squeezing-type terms. As we show later, these
terms enable computational power beyond \BQP.

\emph{Simulation of inertially coupled bosons.---}The problem that we consider is to simulate time evolution for a bosonic system of size $M=2^m$ bosons with sufficiently sparse inertial coupling for a time $t = \poly(m)$. The task is to estimate sums of squared expectation values $\sum_O \abs{\braket{O(t)}}^2$ of time-evolved operators, with $O$ from a given class of $R$-local (i.e., $R$-body) observables, where at most $\poly(m)$  of the initial values $\braket{q_j(0)}$ and $\braket{p_j(0)}$, as well as of the higher moments up to $R^\text{th}$ order, are non-zero. 

We prove that this problem is \BQP-complete. To this end, we follow an encoding scheme that has been used for classical oscillators \cite{Babbush_2023} and more recently for shadow Hamiltonian simulation \cite{somma2025shadowhamiltoniansimulation}. Encoding expectation values of a suitable set of observables $O_a$ as a quantum state via $\ket{\psi} = \sum_a \braket{O_a} \ket{a}$ allows for a simulation in the Heisenberg picture by
integrating the von Neumann equation. Although the resulting effective time
propagator for $\ket\psi$ does not need to be unitary, through the right choice
of the operators $O_a$ defining the encoding one can achieve unitary  
dynamics. We show how to obtain such a Hamiltonian formulation for the class of
inertially coupled bosonic systems by choosing $O_a$ as $R$-local correlation
functions of the canonical operators $q_j$ and $p_j$, whose dynamics are closed under quadratic
Hamiltonians. This allows for an exponential compression of the data for $M=2^m$ bosonic modes into
$O(m)$ qubits.

We denote all canonical operators by $z = (q_1, ..., q_M, p_1, ..., p_M)^T$. We index the canonical operators collectively by $a$, and distinguish position and momentum operators by $j$ and $\bar j$, i.e. $z_j = q_j$ and $z_{\bar j} = p_j$. The time evolution of such a system is determined by the von Neumann equation, which can be written in the Heisenberg picture as
\begin{align}
    \partial_t z = i \comm{H}{z} = \begin{pmatrix} 0 & C \\ -A & 0 \end{pmatrix} z.
    \label{eq:HamiltonEQ}
\end{align}
In the following, we choose a redundant encoding in a Hilbert space of dimension $\mathcal{O}(M^{2R})$ in the spirit of the encoding of classical oscillators in Ref.~\cite{Babbush_2023}. To this end, we introduce the matrices $B,D \in \Mat\left( M \times \binom{M+1}{2}, \mathds{R} \right)$ 
which satisfy $A = B B^\dagger$ and $C = D D^\dagger$. As was done in \cite{Babbush_2023}, $B$ can be chosen as a map between each vertex and the edges incident to it given by the interaction graph:
\begin{align}
    B = &\sum_{j<k} \sqrt{-A_{jk}} (\ket{j} - \ket{k}) \bra{j,k} \nonumber \\
    &+ \sum_j \sqrt{A_{jj} + \sum_{k \neq j} A_{jk}} \ket{j} \bra{j,j},
    \label{eq:B}
\end{align}
and $D$ analogously with $A \leftrightarrow C$. We introduce new coordinates via $z_{\kappa} = i \sum_j B^\dagger_{\kappa j} q_j$ and $z_{\bar \kappa} = \sum_j D^\dagger_{\kappa j} p_j$, which we will distinguish from the original mode indices $z_k, z_{\bar k}$ by greek subscripts. Alternatively, we use $\alpha \in~\left\{1, ..., \binom{M+1}{2}, \bar 1, ..., \overline{\binom{M+1}{2}} \right\}$ to denote both $(j,k)$ and $(\bar j, \bar k)$ with $j \leq k$. This redundant encoding only leads to constant factors in memory, while allowing for a reconstruction of the original degrees of freedom~$z_a$.

Now consider the following decision problem.
\begin{problem}[Simulation of Inertially Coupled Bosons]
\label{prob:QO}
An instance of the problem consists of: 

(i) Oracle access to a $d$-sparse Hamiltonian $H$ of the form
Eq.~\eqref{eq:QO} acting on $M$ bosonic modes (where we can query the
non-zero entries of any row of $H$) with bounded entries $H_{\alpha \beta}, \abs{H_{\alpha \beta}} = \polylog(M)$.

(ii) A list of $\mathcal{O}(\polylog(M))$  nonzero initial moments
\begin{align}
Z_{\alpha_1 \ldots \alpha_r} := \langle z_{\alpha_1}(0) \cdots z_{\alpha_r}(0)\rangle,
\quad 1 \le r \le R,
\end{align}
with $R=\mathcal{O}(1)$, each specified to polynomial precision and satisfying $N^2 := \sum_{r=1}^R \sum_{\alpha_1 ... \alpha_r} \abs{ Z_{\alpha_1 \ldots \alpha_r} }^2 = \polylog(M)$.

(iii) A population $\zeta(I_0)$ over an efficiently characterizable subset $I_0$:
\begin{align}
    \zeta(I_0; t) = \sum_{r=1}^R \sum_{(\alpha_1, \ldots, \alpha_r) \in I_0} \abs{\braket{\prod_{i=1}^r z_{\alpha_i}(t)}}^2,
    \label{eq:read_out}
\end{align}

(iv) A final time $t_f=\mathcal{O}(\polylog(M))$ and real numbers $a > b$ satisfying $a - b = \Omega(1/\polylog(M))$.

The operators evolve according to the Heisenberg equations of motion generated by $H$. The task is to decide whether $\zeta(I_0; t_f) > a$ or $\zeta(I_0; t_f) < b$, promised that one of the two cases holds.
\end{problem}
The restriction on the number of nonzero moments ensures that the input admits a polynomial-size classical description, and is equivalent to assuming access to a polynomial-size quantum circuit preparing the relevant initial state.

Our first result is as follows:

\begin{theorem} \label{thm:BQP}
    The variant of problem \ref{prob:QO} with $E=0$ and oracle access to sparse $B, D$, such that $A=BB^\dagger$ and $C = DD^\dagger$ is \BQP-complete.
\end{theorem}
\BQP-hardness follows immediately since the classical oscillator dynamics of \cite{Babbush_2023} arises as the $R=1$ sector of Problem \ref{prob:QO}. The equations of motion for the first moments $\braket{q_j}$ and $\braket{p_j}$ (corresponding to $R=1$) of a quadratic quantum Hamiltonian are identical to those of the corresponding classical Hamiltonian \eqref{eq:HOHam}. For a diagonal mass matrix $C$, this describes a system of masses coupled by springs, which has been shown to be \BQP-hard \cite{Babbush_2023}.

To prove containment in \BQP, we proceed in three steps: \textit{(i)} Since
$R$-local observables are closed under the dynamics governed by quadratic
Hamiltonians, we first encode the initial values, $\braket{z_{\alpha_1}(0) \cdots z_{\alpha_r}(0)}$, into a quantum state $\ket\psi$, which is efficiently preparable given a polynomial description of the initial moments. \textit{(ii)} We construct an effective Hamiltonian $H_0$ which governs the time evolution of $\ket\psi$. It is sparse and can be efficiently simulated on a quantum computer, using for instance a block-encoding of $H_0$. We use oracle access to $B$ and $D$, for instance through the construction from positive coupling constants discussed in Eq.~(\ref{eq:B}). However, oracle access to $B$ and $D$ is not a necessary condition for \BQP-completeness: For general $A$ and $C$ PSD, a quantum phase estimation based proof of \BQP-containment can be carried out along the lines of Theorem 4 of \cite{Babbush_2023}.
\textit{(iii)} The readout of the form of Eq.~\eqref{eq:read_out} can be estimated by measurements in the computational basis.

\paragraph*{(i) Encoding.} 
In order to simulate quantum oscillators in the Heisenberg picture, we encode the expectation values of all products of up to $R$ generalized coordinates $z_\alpha$ into a quantum state\footnote{
The information about the same generalized coordinates $z_a$ is stored in multiple basis states $\ket{\alpha_1 ... \alpha_r}$ by redundancy in $B$ and $D$. There are ways to reduce memory by constant factors, but we stick to this encoding for the sake of clarity.
}
\begin{align}
    \ket{\psi(t)} = \frac{1}{N} \sum_{r=1}^R \sum_{\alpha_1, ..., \alpha_r = 1}^{2\binom{M+1}{2}} \braket{\prod_{i=1}^r z_{\alpha_i}(t)} \ket{\alpha_1 ... \alpha_r}\ ,
    \label{eq:encoding}
\end{align}
which contains all expectation values of $R$-local monomials in the generalized coordinates at a given time $t$ (which will allow for the simulation of their
dynamics)\footnote{
We can directly read off the coordinate $z_a$ with $a=j$ ($a=\bar j$) from the amplitude $\braket{z_{(a,a)}}$, given that its interaction to the wall given by $k_{jj}$ ($l_{jj}$) is nonzero. Otherwise, the measurement of the relative coordinates $z_{j,k}$ ($z_{\bar j, \bar k}$) is necessary to access $z_a$, see Appendix~\ref{app:measurement}.}. $\ket{\psi(t)}$ can be encoded using $\log\left(
\sum_{r=1}^R \left( 2\binom{M+1}{2}  \right)^r \right) = \mathcal{O}(R \log(M))$
qubits using amplitude encoding. Such a state can be efficiently prepared for polynomially many nonzero amplitudes \cite{Gleinig_2021, ramacciotti2023simplequantumalgorithmefficiently}, such as for the initial state in Problem \ref{prob:QO}.

\paragraph*{(ii) Hamiltonian Encoding of Dynamics.}
The dynamics of the amplitudes $\braket{\prod_{i=1}^r z_{\alpha_i}(t)}$ are governed by the von Neumann equation
\eqref{eq:HamiltonEQ},
which, using the Leibniz rule for derivations, 
straightforwardly evaluates to 
\begin{align}
    \partial_t \braket{z_{\alpha_1} ... z_{\alpha_r}} &= \sum_{j=1}^r \braket{ z_{\alpha_1} \cdots z_{\alpha_{j-1}} (i H_0 z)_{\alpha_j} z_{\alpha_{j+1}} \cdots z_{\alpha_r} } \nonumber \\
    &\text{with} \quad H_0 = \begin{pmatrix}0 & B^\dagger D \\ D^\dagger B & 0\end{pmatrix}.
    \label{eq:eom_k}
\end{align}
Written more compactly, this is an equation for the rank-$r$ tensor 
\begin{align}
    \partial_t \braket{z^{\otimes r}} = i \sum_j H_0^{(j)} \braket{z^{\otimes r}},
    \label{eq:eom_k_tensor}
\end{align}
using the shorthand $H_0^{(j)}$, which has support only on the $j^\text{th}$ tensor factor. Since for each $r=1,\dots,R$, Eq.~\eqref{eq:eom_k_tensor} describes a Hamiltonian evolution of the corresponding sector of $\ket{\psi(t)}$, Eq.~\eqref{eq:encoding}, the encoded state $\ket{\psi(t)}$ evolves according to a Schrödinger equation.

Starting from an oracle for the $d$-sparse coupling matrices $A$ and $C$, one can build $d$-sparse oracles for $B, B^\dagger$ and $D, D^\dagger$ \cite{zschetzsche2025directequivalencedynamicsquantum}. With this, it is straightforward to construct an oracle for the $rd$-sparse Hamiltonian on the right hand side of Eq.~\eqref{eq:eom_k_tensor}. Evolution under sparse Hamiltonians can be efficiently simulated, for instance using one of the methods in \cite{Berry_2015, Berry_2015b,Low_2017, Low_2019}. In direct analogy to \cite{Babbush_2023}, we can ensure a final error $\epsilon$ using $Q = \mathcal{O}(t K d \norm{H_0}_{\rm max} + \log(1/\epsilon))$ queries of $H_0$ and a number of gates $G = \mathcal{O}(Q \log^2(\frac{M Q}{\epsilon}))$.

\paragraph*{(iii) Measurement.}
The read-out of the problem is the estimation of the sum of the squared amplitudes 
\begin{align}
    \abs{\braket{\prod_{i=1}^r z_{\alpha_i}}}^2 = \abs{\braket{\alpha_1 ... \alpha_r | \psi(t)}}^2 N^2
\end{align}
over $(\alpha_1, ..., \alpha_r) \in I_0$ for an efficiently characterizable index set $I_0$. The normalization factor $N$ from Eq.~\eqref{eq:encoding} is set by the initial conditions and remains unchanged under unitary time evolution.

We use an ancilla qubit to encode those basis states that belong to $I_0$ via 
\begin{align}
    \ket{\alpha_1, ..., \alpha_r} \ket{0} \mapsto \ket{\alpha_1, ..., \alpha_r} \ket{f(\alpha_1, ..., \alpha_r)}
\end{align}
with $f(\alpha_1, ..., \alpha_r) = 1$ if $(\alpha_1, ..., \alpha_r) \in I_0$ and 0 else. Finally, the probability to measure the ancilla qubit in the state $\ket{1}$ yields the desired estimate $\frac{1}{N^2} \zeta(I_0; t)$ up to polynomial precision, since $N = \polylog(M)$, thus solving problem \ref{prob:QO}. It is also possible to measure the reconstructed amplitudes $\braket{z_a}$ including their sign. We discuss this in appendix \ref{app:measurement}.

\emph{Quantum Walks are Inertially Coupled Bosons.---}We now relate inertially coupled bosonic systems to continuous-time quantum walks, a canonical model of quantum computation known to be \BQP-complete \cite{Childs_2003, Childs_2009}. It was shown in \cite{zschetzsche2025directequivalencedynamicsquantum} that classical oscillator networks can directly simulate such quantum walks, thereby recovering the \BQP-hardness result of \cite{Babbush_2023}. Both descriptions naturally fall within the class of Hamiltonians for inertially coupled bosons.

Consider a quantum walk on a graph with $M = 2^m$ vertices and adjacency matrix
$T$. The walk corresponds to a single excitation evolving under the
number-preserving bosonic Hamiltonian $H = - \sum_{j,k} T_{jk} a_j^\dagger a_k$.
Within the single excitation sector, replacing $T\;\rightsquigarrow\;
T-c\openone=:-\tilde T$ amounts to a constant energy shift $-c\sum a_j^\dagger
a_j$, which only contributes an unobservable global phase to the evolution; the
Hamiltonian $H=\sum_{j,k} \tilde T_{jk} a_j^\dagger a_k$  thus describes an
equivalent dynamics. If we choose  $c$ at least the maximum vertex degree, 
$\tilde T$  is PSD. 
Substituting $a_j = \tfrac{1}{\sqrt{2}}(q_j + i p_j)$ and $a_j^\dagger = \tfrac{1}{\sqrt{2}}(q_j - i p_j)$, we then obtain
\begin{align}
    H &= \frac{1}{2} \sum_{j,k} \tilde T_{jk} \left( q_j - i p_j \right) \left( q_k + i p_k \right) \nonumber \\
    &= \frac{1}{2} \sum_{j,k} \tilde T_{jk} \left( q_j q_k + p_j p_k + i (q_j p_k - p_j q_k) \right)\:.
    \label{eq:quantum_walk_pre}
\end{align}
As $\tilde T$ is symmetric, and the mixed term is antisymmetric under $j \leftrightarrow k$, 
it vanishes upon summation up to an irrelevant constant
$-\tfrac{1}{2}\sum_j \tilde T_{jj}\openone$. We thus obtain
\begin{align}
    H &= \frac{1}{2} \sum_{j,k} \tilde T_{jk} \left( q_j q_k + p_j p_k \right)\ .
    \label{eq:quantum_walk}
\end{align}
This is precisely the Hamiltonian of inertially coupled bosons, Eq.~\eqref{eq:HOHam}, with positive semidefinite coupling matrices $A = C = \tilde T$. 

Note that $\tilde T$, and thus $A$ and $C$, have precisely the structure required for a network of harmonically coupled oscillators: Its off-diagonal elements are non-positive,
\begin{align}
    A_{jk} = C_{jk} = -k_{jk}\le0 \quad (j \ne k)\ ,
    \label{eq:sign_convention_I}
\end{align}
while the diagonal entries satisfy
\begin{align}
    A_{jj} = C_{jj} = k_{jj} + \sum_{k \ne j} k_{jk}\ ,
    \label{eq:sign_convention_II}
\end{align}
with $k_{jk} \ge 0$ interpreted as spring constants between modes $j$ and $k$ making it a Laplacian-type stiffness matrix. 

More generally, number-preserving hopping Hamiltonians with complex amplitudes yield additional antisymmetric terms of the form $q_j p_k - q_k p_j$ in the $(q,p)$ representation, which yields another class of Hamiltonians for which simulation is \BQP-complete.
\begin{theorem}
\label{thm:hopping}
    The variant of problem~\ref{prob:QO} with $A = C$ and $E^+ = 0$ is \BQP-complete.
\end{theorem}
Hardness is given by the reduction to quantum walks, which are known to be \BQP-complete \cite{Childs_2009}, and containment can be shown analogously to inertially coupled bosons using the coordinates $(iq, p)$, cf.\ Appendix \ref{app:coordinates}.

\emph{Simulating general quadratic Hamiltonians is \PostBQP-complete.---}We now show that
extending the bosonic Hamiltonian beyond the \BQP-complete subclass identified
above fundamentally increases the computational power of the dynamics, elevating
the simulation problem to become \PostBQP-complete (that is, quantum circuits 
augmented by the possibility to project onto specific measurement outcomes).
The key distinction is that more general quadratic Hamiltonians, beyond those considered up to now, induce a linear evolution on the first moments which is no
longer norm-preserving. Such terms enable the exponential amplification or
suppression of selected components, which in turn can be used to emulate quantum
circuits with postselection.  

\begin{theorem} \label{thm:PostBQP}
    The variant of problem \ref{prob:QO} in which $H = \frac{1}{2} \sum_{j,k} H_{jk} z_j z_k$ is an arbitrary, sparse and efficiently queriable quadratic Hamiltonian is \PostBQP-complete.
\end{theorem}
The hardness proof, given in Appendix~\ref{app:PostBQP}, is based on the Feynman–Kitaev construction, as also used to establish \BQP-completeness for classical
oscillators~\cite{Babbush_2023}. However, by additionally introducing squeezing
terms of the form $q_j p_j$ via $E^+$ in Eq.~\eqref{eq:QO}, we are able to implement an exponential suppression of the computational branch corresponding to an output qubit in
state $\ket{0_P}$ relative to $\ket{1_P}$, thereby approximately realizing
the desired postselection on the state $\ket{1_P}$ at the end of the circuit. 

Other families of Hamiltonians with $E^+ = 0$ can also implement postselection. As we show in Appendix~\ref{app:alt_PostBQP}, negative eigenvalues in $A$ or mixed terms from $E^- \neq 0$ can be used to create the same dynamics as in the aforementioned family with $E^+ \neq 0$.

Containment in \PostBQP, or \PP \ equivalently \cite{aaronson2004quantumcomputingpostselectionprobabilistic}, is given in Appendix~\ref{app:PP} and follows a standard reduction of the output as a \GapP \ function.

\emph{Conclusions.---}We have studied the complexity of simulating the time
evolution of quadratic bosonic Hamiltonians, acting on an exponentially large
number of modes $M=2^m$, when restricted to moments of bounded order $R$.  We
have identified two regimes: On the one side, we have introduced the class of
inertially coupled bosons, that is, systems without coupling between positions
and momenta, and shown that simulating its dynamics is $\BQP$-complete, and that
this class encompasses the known \BQP-complete problems of simulating quantum walks and 
classical harmonic oscillators as special cases, thereby extending a previous direct mapping on a level of differential equations \cite{zschetzsche2025directequivalencedynamicsquantum}. We have also shown how to generalize this class such as to encompass all hopping Hamiltonians while
retaining \BQP-completeness. On the other side, we found that the fully general family of quadratic bosonic Hamiltonians, encompassing arbitrary position-momentum couplings beyond
inertial coupling and hopping, gives rise to a sudden jump in complexity, making the problem \PostBQP-complete. The corresponding classical complexity class \PP, which equals \PostBQP, naturally captures the complexity of simulating arbitrary exponentially large bosonic quadratic
Hamiltonians, regardless of the family. It would be interesting to investigate this apparent vanishing of the difference in complexity of classical and quantum simulation more closely. Whether every extension beyond the \BQP-complete subclasses identified above already exhibits this transition, or whether there exist quadratic bosonic Hamiltonian families of intermediate complexity, remains open.

A noteworthy point is the role played by squeezing in this complexity transition. On the one hand, we find that all hopping Hamiltonians (that is, those without squeezing terms) are inside \BQP. On the other hand, the presence of a squeezing term does not immediately imply \PostBQP-hardness, as can be seen by e.g.\ setting $A\ne C$ in Eq.~\eqref{eq:HOHam}. The reason for this can be intuitively grasped already for a single mode: The presence of a squeezing term in the Hamiltonian does not necessarily imply that squeezing builds up over time, as rotating terms can lead to an averaging out of the squeezing. Let us remark that a similar effect of squeezing on complexity has been observed in the context of simulating photonic circuits with particle number-conserving gates: As long as only a bounded number of squeezing operations is allowed, the problem is \BQP-complete~\cite{Barthe_2025}, while it becomes harder in the presence of squeezing~\cite{Alice2025}. We emphasize that the use of postselection in Ref.~\cite{Barthe_2025} is algorithmic: there, squeezing gives rise to an effective imaginary-time evolution implemented via a heralded, postselection-based circuit construction, i.e.\ postselection is a tool for \emph{simulating} squeezing dynamics. The connection we establish runs in the converse direction: squeezing-type terms in the physical Hamiltonian itself render the simulation problem \PostBQP-hard, i.e.\ at least as powerful as arbitrary postselected computation.

\emph{Acknowledgments.---}We acknowledge support from the Austrian Science Fund FWF (Grant No.\ \href{https://doi.org/10.55776/COE1}{10.55776/COE1} and \href{https://doi.org/10.55776/F71}{10.55776/F71}), by the European Union -- NextGenerationEU, and by the European Union’s Horizon 2020 research and innovation programme through Grant No.\ 863476 (ERC-CoG \mbox{SEQUAM}).

\bibliographystyle{apsrev4-2}
\bibliography{literature}

\begin{thebibliography}{18}%
\makeatletter
\providecommand \@ifxundefined [1]{%
 \@ifx{#1\undefined}
}%
\providecommand \@ifnum [1]{%
 \ifnum #1\expandafter \@firstoftwo
 \else \expandafter \@secondoftwo
 \fi
}%
\providecommand \@ifx [1]{%
 \ifx #1\expandafter \@firstoftwo
 \else \expandafter \@secondoftwo
 \fi
}%
\providecommand \natexlab [1]{#1}%
\providecommand \enquote  [1]{``#1''}%
\providecommand \bibnamefont  [1]{#1}%
\providecommand \bibfnamefont [1]{#1}%
\providecommand \citenamefont [1]{#1}%
\providecommand \href@noop [0]{\@secondoftwo}%
\providecommand \href [0]{\begingroup \@sanitize@url \@href}%
\providecommand \@href[1]{\@@startlink{#1}\@@href}%
\providecommand \@@href[1]{\endgroup#1\@@endlink}%
\providecommand \@sanitize@url [0]{\catcode `\\12\catcode `\$12\catcode
  `\&12\catcode `\#12\catcode `\^12\catcode `\_12\catcode `\%12\relax}%
\providecommand \@@startlink[1]{}%
\providecommand \@@endlink[0]{}%
\providecommand \url  [0]{\begingroup\@sanitize@url \@url }%
\providecommand \@url [1]{\endgroup\@href {#1}{\urlprefix }}%
\providecommand \urlprefix  [0]{URL }%
\providecommand \Eprint [0]{\href }%
\providecommand \doibase [0]{https://doi.org/}%
\providecommand \selectlanguage [0]{\@gobble}%
\providecommand \bibinfo  [0]{\@secondoftwo}%
\providecommand \bibfield  [0]{\@secondoftwo}%
\providecommand \translation [1]{[#1]}%
\providecommand \BibitemOpen [0]{}%
\providecommand \bibitemStop [0]{}%
\providecommand \bibitemNoStop [0]{.\EOS\space}%
\providecommand \EOS [0]{\spacefactor3000\relax}%
\providecommand \BibitemShut  [1]{\csname bibitem#1\endcsname}%
\let\auto@bib@innerbib\@empty
\bibitem [{\citenamefont {Kraus}(2011)}]{Kraus2011}%
  \BibitemOpen
  \bibfield  {author} {\bibinfo {author} {\bibfnamefont {B.}~\bibnamefont
  {Kraus}},\ }\href {https://doi.org/10.1103/PhysRevLett.107.250503} {\bibfield
   {journal} {\bibinfo  {journal} {Phys. Rev. Lett.}\ }\textbf {\bibinfo
  {volume} {107}},\ \bibinfo {pages} {250503} (\bibinfo {year}
  {2011})}\BibitemShut {NoStop}%
\bibitem [{\citenamefont {Babbush}\ \emph {et~al.}(2023)\citenamefont
  {Babbush}, \citenamefont {Berry}, \citenamefont {Kothari}, \citenamefont
  {Somma},\ and\ \citenamefont {Wiebe}}]{Babbush_2023}%
  \BibitemOpen
  \bibfield  {author} {\bibinfo {author} {\bibfnamefont {R.}~\bibnamefont
  {Babbush}}, \bibinfo {author} {\bibfnamefont {D.~W.}\ \bibnamefont {Berry}},
  \bibinfo {author} {\bibfnamefont {R.}~\bibnamefont {Kothari}}, \bibinfo
  {author} {\bibfnamefont {R.~D.}\ \bibnamefont {Somma}},\ and\ \bibinfo
  {author} {\bibfnamefont {N.}~\bibnamefont {Wiebe}},\ }\bibfield  {journal}
  {\bibinfo  {journal} {Physical Review X}\ }\textbf {\bibinfo {volume} {13}},\
  \href {https://doi.org/10.1103/physrevx.13.041041}
  {10.1103/physrevx.13.041041} (\bibinfo {year} {2023})\BibitemShut {NoStop}%
\bibitem [{\citenamefont {Childs}\ \emph {et~al.}(2003)\citenamefont {Childs},
  \citenamefont {Cleve}, \citenamefont {Deotto}, \citenamefont {Farhi},
  \citenamefont {Gutmann},\ and\ \citenamefont {Spielman}}]{Childs_2003}%
  \BibitemOpen
  \bibfield  {author} {\bibinfo {author} {\bibfnamefont {A.~M.}\ \bibnamefont
  {Childs}}, \bibinfo {author} {\bibfnamefont {R.}~\bibnamefont {Cleve}},
  \bibinfo {author} {\bibfnamefont {E.}~\bibnamefont {Deotto}}, \bibinfo
  {author} {\bibfnamefont {E.}~\bibnamefont {Farhi}}, \bibinfo {author}
  {\bibfnamefont {S.}~\bibnamefont {Gutmann}},\ and\ \bibinfo {author}
  {\bibfnamefont {D.~A.}\ \bibnamefont {Spielman}},\ }in\ \href
  {https://doi.org/10.1145/780542.780552} {\emph {\bibinfo {booktitle}
  {Proceedings of the thirty-fifth annual ACM symposium on Theory of
  computing}}},\ \bibinfo {series and number} {STOC03}\ (\bibinfo  {publisher}
  {ACM},\ \bibinfo {year} {2003})\BibitemShut {NoStop}%
\bibitem [{\citenamefont {Childs}(2009)}]{Childs_2009}%
  \BibitemOpen
  \bibfield  {author} {\bibinfo {author} {\bibfnamefont {A.~M.}\ \bibnamefont
  {Childs}},\ }\bibfield  {journal} {\bibinfo  {journal} {Physical Review
  Letters}\ }\textbf {\bibinfo {volume} {102}},\ \href
  {https://doi.org/10.1103/physrevlett.102.180501}
  {10.1103/physrevlett.102.180501} (\bibinfo {year} {2009})\BibitemShut
  {NoStop}%
\bibitem [{\citenamefont {Somma}\ \emph {et~al.}(2025)\citenamefont {Somma},
  \citenamefont {King}, \citenamefont {Kothari}, \citenamefont {O'Brien},\ and\
  \citenamefont {Babbush}}]{somma2025shadowhamiltoniansimulation}%
  \BibitemOpen
  \bibfield  {author} {\bibinfo {author} {\bibfnamefont {R.~D.}\ \bibnamefont
  {Somma}}, \bibinfo {author} {\bibfnamefont {R.}~\bibnamefont {King}},
  \bibinfo {author} {\bibfnamefont {R.}~\bibnamefont {Kothari}}, \bibinfo
  {author} {\bibfnamefont {T.}~\bibnamefont {O'Brien}},\ and\ \bibinfo {author}
  {\bibfnamefont {R.}~\bibnamefont {Babbush}},\ }\href
  {https://arxiv.org/abs/2407.21775} {\bibinfo {title} {Shadow hamiltonian
  simulation}} (\bibinfo {year} {2025}),\ \Eprint
  {https://arxiv.org/abs/2407.21775} {arXiv:2407.21775 [quant-ph]} \BibitemShut
  {NoStop}%
\bibitem [{\citenamefont {Gleinig}\ and\ \citenamefont
  {Hoefler}(2021)}]{Gleinig_2021}%
  \BibitemOpen
  \bibfield  {author} {\bibinfo {author} {\bibfnamefont {N.}~\bibnamefont
  {Gleinig}}\ and\ \bibinfo {author} {\bibfnamefont {T.}~\bibnamefont
  {Hoefler}},\ }in\ \href {https://doi.org/10.1109/DAC18074.2021.9586240}
  {\emph {\bibinfo {booktitle} {2021 58th ACM/IEEE Design Automation Conference
  (DAC)}}}\ (\bibinfo {year} {2021})\ pp.\ \bibinfo {pages}
  {433--438}\BibitemShut {NoStop}%
\bibitem [{\citenamefont {Ramacciotti}\ \emph {et~al.}(2023)\citenamefont
  {Ramacciotti}, \citenamefont {Lefterovici},\ and\ \citenamefont
  {Rotundo}}]{ramacciotti2023simplequantumalgorithmefficiently}%
  \BibitemOpen
  \bibfield  {author} {\bibinfo {author} {\bibfnamefont {D.}~\bibnamefont
  {Ramacciotti}}, \bibinfo {author} {\bibfnamefont {A.-I.}\ \bibnamefont
  {Lefterovici}},\ and\ \bibinfo {author} {\bibfnamefont {A.~F.}\ \bibnamefont
  {Rotundo}},\ }\href {https://arxiv.org/abs/2310.19309} {\bibinfo {title} {A
  simple quantum algorithm to efficiently prepare sparse states}} (\bibinfo
  {year} {2023}),\ \Eprint {https://arxiv.org/abs/2310.19309} {arXiv:2310.19309
  [quant-ph]} \BibitemShut {NoStop}%
\bibitem [{\citenamefont {Zschetzsche}\ \emph {et~al.}(2025)\citenamefont
  {Zschetzsche}, \citenamefont {Mansuroglu}, \citenamefont {Molnár},\ and\
  \citenamefont {Schuch}}]{zschetzsche2025directequivalencedynamicsquantum}%
  \BibitemOpen
  \bibfield  {author} {\bibinfo {author} {\bibfnamefont {L.}~\bibnamefont
  {Zschetzsche}}, \bibinfo {author} {\bibfnamefont {R.}~\bibnamefont
  {Mansuroglu}}, \bibinfo {author} {\bibfnamefont {A.}~\bibnamefont
  {Molnár}},\ and\ \bibinfo {author} {\bibfnamefont {N.}~\bibnamefont
  {Schuch}},\ }\href {https://arxiv.org/abs/2512.03681} {\bibinfo {title}
  {Direct equivalence between dynamics of quantum walks and coupled classical
  oscillators}} (\bibinfo {year} {2025}),\ \Eprint
  {https://arxiv.org/abs/2512.03681} {arXiv:2512.03681 [quant-ph]} \BibitemShut
  {NoStop}%
\bibitem [{\citenamefont {Berry}\ \emph
  {et~al.}(2015{\natexlab{a}})\citenamefont {Berry}, \citenamefont {Childs},
  \citenamefont {Cleve}, \citenamefont {Kothari},\ and\ \citenamefont
  {Somma}}]{Berry_2015}%
  \BibitemOpen
  \bibfield  {author} {\bibinfo {author} {\bibfnamefont {D.~W.}\ \bibnamefont
  {Berry}}, \bibinfo {author} {\bibfnamefont {A.~M.}\ \bibnamefont {Childs}},
  \bibinfo {author} {\bibfnamefont {R.}~\bibnamefont {Cleve}}, \bibinfo
  {author} {\bibfnamefont {R.}~\bibnamefont {Kothari}},\ and\ \bibinfo {author}
  {\bibfnamefont {R.~D.}\ \bibnamefont {Somma}},\ }\bibfield  {journal}
  {\bibinfo  {journal} {Physical Review Letters}\ }\textbf {\bibinfo {volume}
  {114}},\ \href {https://doi.org/10.1103/physrevlett.114.090502}
  {10.1103/physrevlett.114.090502} (\bibinfo {year}
  {2015}{\natexlab{a}})\BibitemShut {NoStop}%
\bibitem [{\citenamefont {Berry}\ \emph
  {et~al.}(2015{\natexlab{b}})\citenamefont {Berry}, \citenamefont {Childs},\
  and\ \citenamefont {Kothari}}]{Berry_2015b}%
  \BibitemOpen
  \bibfield  {author} {\bibinfo {author} {\bibfnamefont {D.~W.}\ \bibnamefont
  {Berry}}, \bibinfo {author} {\bibfnamefont {A.~M.}\ \bibnamefont {Childs}},\
  and\ \bibinfo {author} {\bibfnamefont {R.}~\bibnamefont {Kothari}},\ }in\
  \href {https://doi.org/10.1109/focs.2015.54} {\emph {\bibinfo {booktitle}
  {2015 IEEE 56th Annual Symposium on Foundations of Computer Science}}}\
  (\bibinfo  {publisher} {IEEE},\ \bibinfo {year} {2015})\ p.\ \bibinfo {pages}
  {792–809}\BibitemShut {NoStop}%
\bibitem [{\citenamefont {Low}\ and\ \citenamefont {Chuang}(2017)}]{Low_2017}%
  \BibitemOpen
  \bibfield  {author} {\bibinfo {author} {\bibfnamefont {G.~H.}\ \bibnamefont
  {Low}}\ and\ \bibinfo {author} {\bibfnamefont {I.~L.}\ \bibnamefont
  {Chuang}},\ }\bibfield  {journal} {\bibinfo  {journal} {Physical Review
  Letters}\ }\textbf {\bibinfo {volume} {118}},\ \href
  {https://doi.org/10.1103/physrevlett.118.010501}
  {10.1103/physrevlett.118.010501} (\bibinfo {year} {2017})\BibitemShut
  {NoStop}%
\bibitem [{\citenamefont {Low}\ and\ \citenamefont {Chuang}(2019)}]{Low_2019}%
  \BibitemOpen
  \bibfield  {author} {\bibinfo {author} {\bibfnamefont {G.~H.}\ \bibnamefont
  {Low}}\ and\ \bibinfo {author} {\bibfnamefont {I.~L.}\ \bibnamefont
  {Chuang}},\ }\href {https://doi.org/10.22331/q-2019-07-12-163} {\bibfield
  {journal} {\bibinfo  {journal} {Quantum}\ }\textbf {\bibinfo {volume} {3}},\
  \bibinfo {pages} {163} (\bibinfo {year} {2019})}\BibitemShut {NoStop}%
\bibitem [{\citenamefont
  {Aaronson}(2004)}]{aaronson2004quantumcomputingpostselectionprobabilistic}%
  \BibitemOpen
  \bibfield  {author} {\bibinfo {author} {\bibfnamefont {S.}~\bibnamefont
  {Aaronson}},\ }\href {https://arxiv.org/abs/quant-ph/0412187} {\bibinfo
  {title} {Quantum computing, postselection, and probabilistic
  polynomial-time}} (\bibinfo {year} {2004}),\ \Eprint
  {https://arxiv.org/abs/quant-ph/0412187} {arXiv:quant-ph/0412187 [quant-ph]}
  \BibitemShut {NoStop}%
\bibitem [{\citenamefont {Barthe}\ \emph {et~al.}(2025)\citenamefont {Barthe},
  \citenamefont {Cerezo}, \citenamefont {Sornborger}, \citenamefont {Larocca},\
  and\ \citenamefont {García-Martín}}]{Barthe_2025}%
  \BibitemOpen
  \bibfield  {author} {\bibinfo {author} {\bibfnamefont {A.}~\bibnamefont
  {Barthe}}, \bibinfo {author} {\bibfnamefont {M.}~\bibnamefont {Cerezo}},
  \bibinfo {author} {\bibfnamefont {A.~T.}\ \bibnamefont {Sornborger}},
  \bibinfo {author} {\bibfnamefont {M.}~\bibnamefont {Larocca}},\ and\ \bibinfo
  {author} {\bibfnamefont {D.}~\bibnamefont {García-Martín}},\ }\bibfield
  {journal} {\bibinfo  {journal} {Physical Review Letters}\ }\textbf {\bibinfo
  {volume} {134}},\ \href {https://doi.org/10.1103/physrevlett.134.070604}
  {10.1103/physrevlett.134.070604} (\bibinfo {year} {2025})\BibitemShut
  {NoStop}%
\bibitem [{\citenamefont {Barthe}(2025)}]{Alice2025}%
  \BibitemOpen
  \bibfield  {author} {\bibinfo {author} {\bibfnamefont {A.}~\bibnamefont
  {Barthe}},\ }\href@noop {} {}\bibinfo {howpublished} {private communication}
  (\bibinfo {year} {2025})\BibitemShut {NoStop}%
\bibitem [{\citenamefont {Simon}(1995)}]{Simon1995RankOne}%
  \BibitemOpen
  \bibfield  {author} {\bibinfo {author} {\bibfnamefont {B.}~\bibnamefont
  {Simon}},\ }in\ \href@noop {} {\emph {\bibinfo {booktitle} {Mathematical
  Quantum Theory I: Field Theory and Many-Body Theory}}},\ \bibinfo {series}
  {CRM Proceedings and Lecture Notes}, Vol.~\bibinfo {volume} {7},\ \bibinfo
  {editor} {edited by\ \bibinfo {editor} {\bibfnamefont {J.}~\bibnamefont
  {Feldman}}, \bibinfo {editor} {\bibfnamefont {J.}~\bibnamefont {Froehlich}},\
  and\ \bibinfo {editor} {\bibfnamefont {V.}~\bibnamefont {Rivasseau}}}\
  (\bibinfo  {publisher} {American Mathematical Society},\ \bibinfo {address}
  {Providence, RI},\ \bibinfo {year} {1995})\ pp.\ \bibinfo {pages}
  {109--149}\BibitemShut {NoStop}%
\bibitem [{\citenamefont {Teschl}(2000)}]{Teschl2000Jacobi}%
  \BibitemOpen
  \bibfield  {author} {\bibinfo {author} {\bibfnamefont {G.}~\bibnamefont
  {Teschl}},\ }\href@noop {} {\emph {\bibinfo {title} {Jacobi Operators and
  Completely Integrable Nonlinear Lattices}}},\ \bibinfo {series} {Mathematical
  Surveys and Monographs}, Vol.~\bibinfo {volume} {72}\ (\bibinfo  {publisher}
  {American Mathematical Society},\ \bibinfo {address} {Providence, RI},\
  \bibinfo {year} {2000})\BibitemShut {NoStop}%
\bibitem [{\citenamefont {Higham}(2008)}]{Higham2008}%
  \BibitemOpen
  \bibfield  {author} {\bibinfo {author} {\bibfnamefont {N.~J.}\ \bibnamefont
  {Higham}},\ }\href {https://doi.org/10.1137/1.9780898717778} {\emph {\bibinfo
  {title} {Functions of Matrices: Theory and Computation}}}\ (\bibinfo
  {publisher} {Society for Industrial and Applied Mathematics (SIAM)},\
  \bibinfo {address} {Philadelphia, PA},\ \bibinfo {year} {2008})\BibitemShut
  {NoStop}%
\end{thebibliography}%

\appendix
\onecolumngrid

\makeatletter
\def\set@footnotewidth{\onecolumngrid}
\def\footnoterule{\kern-6pt\hrule width 1.5in\kern6pt}
\makeatother

\section{Post-Processing of Measurement Data}\label{app:measurement}
In this appendix, we extend the discussion on the read-out of quadrature expectation values and discuss how to reconstruct the expectation values $\braket{z_{a_1} \cdots z_{a_r}}$ from the encoded amplitudes $\braket{z_{\alpha_1} \cdots z_{\alpha_r}}$ and how to improve efficiency including all samples in computational basis measurements using the example of linear and quadratic expectation values. 

Because position coordinates are mixed by $B$ and momentum coordinates by $D$, only specific linear combinations of the generalized coordinates $z_a$ are directly accessible through computational-basis measurement. However, the encoding via the indices $\alpha$ is redundant, and expectation values of individual generalized coordinates can be reconstructed from diagonal components. In particular, the coordinates 
\begin{align}
    \braket{z_{(j,j)}} &= \sum_{k=1}^M B^\dagger_{(j,j) k} \braket{q_k} = \sqrt{ A_{jj} + \sum_{k \neq j} A_{jk} } \, \braket{q_j}
    \label{eq:zjj}
\end{align}
directly encode the position of mode $j$, provided that $A_{jj} + \sum_{k\neq j} A_{jk} \neq 0$. Analogous statements hold for momentum coordinates using $C$ and $D$. Consequently, the absolute value of any $r$-local correlation $\big|\langle z_{a_1}\cdots z_{a_r}\rangle\big|$ can be inferred from the probability of measuring the basis state $\ket{\alpha_1\ldots\alpha_r}
= \ket{(a_1,a_1)\ldots(a_r,a_r)}$, assuming nonvanishing coupling to the reference.

To recover sign information, one can use a superposition of the vacuum component $\ket{\mathds{1}}$ (corresponding to $r=0$) with the basis state $\ket{(a_1,a_1)\ldots(a_r,a_r)}$ obtained through a Hadamard transformation. Measuring in the rotated basis yields probabilities proportional to $\abs{1 \pm \braket{z_{(a_1, a_1)} \cdots z_{(a_r, a_r)}}}^2$, which 
allows one to infer the sign of $\braket{z_{a_1} \cdots z_{a_r}}$ whenever it is polynomially separated from zero.

In practice, one can gain efficiency when estimating all relevant amplitudes by recycling measurement outcomes. As measurements may also yield outcomes corresponding to coordinates $z_{(a,b)}$, which encode relative positions and momenta, these additional samples can be incorporated to reduce estimator variance. We restrict the discussion to the cases $r=1$ and $r=2$ and to positional coordinates $q$, governed by the matrix $A$. The analogous construction for momenta follows by replacing $q \to p$, $A \to C$, and $B \to D$.

\paragraph*{Linear expectation values.}
For $r=1$, we have
\begin{align}
    \braket{z_{(j,j')}} &= \sqrt{- A_{jj'} } \braket{q_j - q_{j'}} \\
    \braket{z_{(j,j)}} &= \sqrt{ A_{jj} + \sum_{k \neq j} A_{jk} } \, \braket{q_j}.
\end{align}
The prefactor $\sqrt{ A_{jj} + \sum_{k \neq j} A_{jk} }$ can be computed using a single row query of $A$. Measuring the probability of the computational basis state $\ket{(j,j)}$ therefore allows to compute
\begin{align}
    \abs{ \braket{q_j} } = \frac{1}{\sqrt{ A_{jj} + \sum_{k \neq j} A_{jk} }} \abs{ \braket{z_{(j,j)}} }\ .
\end{align}
The sign of $\braket{q_j}$ is not accessible from computational-basis measurements alone but can be recovered by interference, for example via a Hadamard transformation mixing $q_j$ with the vacuum component corresponding to $r=0$. A relative sign between two modes can be retrieved mixing $(q_j, q_{j'}) \to (q_j + q_{j'}, q_j - q_{j'})$.

\paragraph*{Quadratic correlations.}
Certain second moments are directly accessible. For example,
\begin{align}
    \abs{\braket{z_{(j,j')}^2}} = \abs{A_{jj'}} \braket{(q_j - q_{j'})^2},
\end{align}
which corresponds to the interaction energy associated with the $(j,j')$ coupling. More generally, cross-correlations $\braket{q_j q_{j'}}$ can be reconstructed provided that both modes have nonvanishing coupling to the reference and to each other. In that case one obtains
\begin{align}
    \braket{q_j q_{j'}} &= \frac{1}{2} \left( \frac{1}{A_{jj'}} \braket{z_{(j,j')}^2} + \frac{1}{A_{jj} + \sum_{k \neq j} A_{jk}} \braket{z_{(j,j)}^2} + \frac{1}{A_{j'j'} + \sum_{k \neq j'} A_{j'k}} \braket{z_{(j',j')}^2} \right), \label{eq:qq}
\end{align}
or equivalently,
\begin{align}
    \braket{q_j q_{j'}} &= \frac{1}{2\sqrt{A_{jj} + \sum_{k \neq j} A_{jk}} \sqrt{A_{j'j'} + \sum_{k \neq j'} A_{j'k}}} \braket{z_{(j,j)} z_{(j',j')}}.
    \label{eq:qq_2}
\end{align}
Since $\braket{z_\alpha^2} \ge 0$, we omit absolute values. These expressions yield the energy expectation values $A_{jj'} \braket{q_j q_{j'}}$, provided $A_{jj} + \sum_{k \neq j} A_{jk} \neq 0$ and similarly for $j'$. Extensions to higher-order products and to momentum coordinates follow analogously.

\section{\PostBQP-hardness of generic Gaussian Hamiltonians (Theorem \ref{thm:PostBQP})}
\label{app:PostBQP}

In this appendix, we prove Theorem~\ref{thm:PostBQP} by showing that simulating
the first-moment dynamics of generic Gaussian Hamiltonians with inertial
couplings and mixed terms from $E^+$ is \PostBQP-hard. The proof proceeds by an explicit
reduction from a \PostBQP-complete promise problem to the evolution of a
quadratic bosonic Hamiltonian. We start by stating the \PostBQP-complete bosonic
problem in subsection 1, and then proceed to analyze it in subsections 2--5.
Subsection 6 contains a technical result required in the analysis, which is
standard but provided for completeness. Finally, subsection 7 discusses
alternative \PostBQP-hard families of bosonic Hamiltonians.

\subsection{PostBQP-completeness Construction}

\PostBQP \ (Postselected Bounded-Error Quantum Polynomial Time) is the class of
decision problems solvable by a polynomial-time quantum algorithm that is
allowed to postselect on the outcome of a measurement
\cite{aaronson2004quantumcomputingpostselectionprobabilistic}. The natural
\PostBQP-complete problem is \textsc{postselected circuit}: Consider a quantum
register consisting of $m$ qubits -- a postselection qubit $P$, an output qubit
$Q$ and $m-2$ qubits in a register $R$ -- and a uniform quantum circuit $U =
\prod_{l=1}^L U_l$, $L=\poly(m)$, on those $m$ qubits, where the $U_l$ are
chosen from the computationally universal gateset
$\{\mathrm{Hadamard},\mathrm{Toffoli}\}$. Let
$\ket{\Psi_1}:=\tfrac{1}{\delta}\langle 1_P\rvert U\ket{0_P 0_Q 0_R}$ be the
normalized state obtained from evolving $\ket{0_P 0_Q 0_R}$ with $U$ and
postselecting on the $P$ register being $\ket{1_P}$. The task is then to decide
whether $\norm{ \braket{ 1_Q | \Psi_1 }}^2$ is greater than $2/3$ or smaller
than $1/3$, promised that one of the two is the case. Importantly, since
superpositions of computational basis states are only created by the Hadamard
gates in the circuit, we have that $\delta^2 = \Omega(2^{-L})$.

In the following, we provide a reduction from \textsc{postselected circuit} to
the simulation of quadratic bosonic Hamiltonians, as defined in Theorem~\ref{thm:PostBQP} in the main text.
The reduction involves three parts: Dynamics, the initial state, and the
read-out. We start by stating the individual parts of the reduction, and
analyze it in the subsequent subsections.

\vspace{0.4em}
\noindent\textit{(i) Dynamics:} We encode the quantum circuit $U$ into a Feynman--Kitaev Hamiltonian
\begin{align}
    H_{\rm FK} = \sum_{l=1}^{L+1} \ket{l} \bra{l+1} \otimes U_l^\dagger + \ket{l+1} \bra{l} \otimes U_l
    \label{eq:FK}
\end{align}
where for technical reasons, we add one extra time step with $U_{L+1}=\mathds{1}$ at the end.
$H_{\mathrm{FK}}$ acts on a tensor product of a ``clock register'' $\ket{l}$, $l=1,\dots,L+2$, and
the $m$ qubit register of the circuit. The idea is that time evolution under
$H_{\mathrm{FK}}$ propagates an initial state $\ket{l=0}\otimes\ket{\phi}$ by
simultaneously propagating the time register and  applying the corresponding
gate $U_l$ to the qubit register, such that the state of the system at any time
is a superposition of states of the form $\ket{l}\otimes(U_l\cdots
U_1\ket{\phi})$. We encode the Hilbert space of Eq.~\eqref{eq:FK} in terms of
$M=(L+2)2^{m}$ bosonic modes, where we label position and momentum operators by
$q^{(l)}_{j_P,j_Q,j_R}$ and $p^{(l)}_{j_P,j_Q,j_R}$, respectively, with a layer index
$l\in\{1, ..., L+2\}$, and qubit indices $j_P,j_Q\in\{0,1\}$ and $j_R \in
\{0,1\}^{\times (m-2)}$. To define the encoding, consider the dynamics of first
moments under the Hamiltonian 
\begin{align}
    H &= \frac{1}{2} \begin{pmatrix} q & p \end{pmatrix} \begin{pmatrix} A & E^+ \\ E^+ & C \end{pmatrix} \begin{pmatrix} q \\ p \end{pmatrix}\ ,
    \label{eq:FK_Hamiltonian}
\end{align}
with
\begin{align}
    A &= 4 \mathds{1} - \sum_l^{L+1} \left( \ket{l} \bra{l+1} \otimes U_l^\dagger + \ket{l+1}\bra{l} \otimes U_l \right)\ , \label{eq:A_FK} \\
    C &= \mathds{1}\ , \\
    E^+ &= -\beta \ket{L+2} \bra{L+2} \otimes \ket{1_P}\bra{1_P}. 
    \label{eq:FK_Hamiltonian-series-end}
\end{align}
Since our gate set consists of real-valued gates, $A$ is a real-symmetric matrix, as required.
In order to match the sign convention (cf. Eqs.~(\ref{eq:sign_convention_I}, \ref{eq:sign_convention_II}) of the main text), one can
additionally encode negative entries in the unitaries $U_l$ through positive
interactions using the sign-separation trick at the cost of doubling the number of
modes (corresponding to adding one ancilla qubit which is put in the
$\ket{-}$ state), see~\cite{Babbush_2023,zschetzsche2025directequivalencedynamicsquantum}. $E^+$ is a
diagonal rank-one term of strength $\beta$ acting only on the final clock
layer. Intuitively, this term acts as an energetic bias favoring computational
histories that terminate with the postselected qubit $P$ in the postselected state
$\ket{1_P}$, while leaving all intermediate layers unchanged.

\vspace{0.4em}
\noindent\textit{(ii) Initial state:} 
The initial conditions of the bosonic
system to which we reduce \textsc{Postselected Circuit} are given by $\big\langle
q_j^{(l)}(0)\big\rangle = 0$ and $\big\langle p_j^{(l)}(0)\big\rangle = 
\big\langle \dot q_j^{(l)}(0)\big\rangle =
\delta_{l,1} \delta_{j,0}$.

\vspace{0.4em}
\noindent\textit{(iii) Read-out:} 
Recall that read-outs involve estimating quantities of the form 
\begin{align}
    \zeta(I;t):=\sum_{r=1}^R \sum_{(\alpha_1, \ldots, \alpha_r) \in I} \abs{\braket{\prod_{i=1}^r z_{\alpha_i}(t)}}^2
\end{align}
over an efficiently characterizable index set $I$,  up to
$1/\mathrm{polylog}(M)=1/\mathrm{poly}(m)$ precision.  For our reduction, we
need to perform two such read-outs after a suitably chosen time $t_f=O(L)$. They are given by the index sets
\begin{equation}
    \label{eq:postbqp-readout-indexsets}
I_{\mathrm{num}} = \{ (l, \alpha) = (\overline{L+1}, \bar 1_P \bar 1_Q \bar j_R) \}
\mbox{\quad and\quad}
I_{\mathrm{den}} = \{ (l, \alpha) = (\overline{L+1}, \bar 1_P \bar b_Q \bar j_R) \}\ ,
\end{equation}
respectively,
each for all possible choices of the labels $\bar b_Q$ and
$\bar j_R$ (recall that the bar denotes momentum labels; that is, the measurement reads out a sum of expectation values of momenta). As we show below, both populations share a common multiplicative prefactor $\delta$ which depends on the (unknown) postselection probability of the circuit $U$ and thus cannot be computed classically in advance. We therefore estimate $\zeta(I_\mathrm{num};t_f)$ and $\zeta(I_\mathrm{den};t_f)$ independently via a sequence of queries to Problem \ref{prob:QO}, and combine the results classically. We accept exactly if
$\zeta(I_\mathrm{num};t_f)/\zeta(I_\mathrm{den};t_f)>1/2$. As we will show in the following, 
for a suitably chosen time $t=O(L)$, this algorithm accepts exactly if the original \textsc{postselected circuit} instance accepts and thus constitutes a polynomial-time Turing reduction from \textsc{postselected circuit} to Problem \ref{prob:QO}.

\subsection{Formal Solution to the Bosonic Problem}

The equations of motion 
defined in Eqs.~(\ref{eq:FK_Hamiltonian}-\ref{eq:FK_Hamiltonian-series-end})
can be rewritten as the following second order linear differential equation in $q$:
\begin{align}
    \big\langle\ddot q_j^{(l)}\big\rangle &= \left[ - (A - (E^+)^2) \braket{q} \right]_j^{(l)} =: \left[- \tilde A \braket{q} \right]_j^{(l)},
    \label{eq:diff_HO}
\end{align}
with $\tilde A := A - (E^+)^2$, which ceases to be positive semidefinite for large enough $\beta$. The momentum coordinates $p_j^{(l)}$ can be transformed into velocities $\dot q_j^{(l)}$ via
\begin{align}
    \big\langle\dot q_j^{(l)}\big\rangle &= \big\langle{p_j^{(l)}}\big\rangle - \delta_{l,L+2} \delta_{j_P, 1} \beta \big\langle{q_j^{(l)}}\big\rangle\ .
    \label{eq:postbqp-qdot-vs-p}
\end{align}
Correspondingly, the initial state 
$\big\langle p_j^{(l)}(0)\big\rangle = \delta_{l,1} \delta_{j,0}$  translates to
$\big\langle \dot q_j^{(l)}(0) \big\rangle =  \delta_{l,1} \delta_{j,0}$, and as
before $\big\langle q_j^{(l)}(0)\big\rangle \!=\!
0$.\footnote{Equation~\eqref{eq:postbqp-qdot-vs-p} is also the reason why we
added an extra time step $U_{L+1}=\mathds{1}$, since this way, at time $l=L+1$,
$\big\langle\dot q_j^{(L+1)}\big\rangle = \big\langle{p_j^{(L+1)}}\big\rangle$.}

The general solution to Eq.~\eqref{eq:diff_HO} reads 
\begin{align}
    \braket{\dot q(t)} + i \sqrt{\tilde A} \braket{q(t)} = e^{it \sqrt{\tilde A}} \left( \braket{\dot q(0)} + i \sqrt{\tilde A} \braket{q(0)} \right),
\end{align}
with initial conditions $\braket{\dot q(0)}$ and $\sqrt{\tilde A} \braket{q(0)}$. Using the initial condition
$\big\langle q_j^{(l)}(0)\big\rangle = 0$, this simplifies to
\begin{align}
        \braket{q(t)} = \sqrt{\tilde A}^{-1} \sin\left( \sqrt{\tilde A} t \right) \braket{\dot q(0)}, \qquad \text{and} \qquad \braket{\dot q(t)} = \cos\left( \sqrt{\tilde A} t \right) \braket{\dot q(0)}\ .
        \label{eq:timeevol-atilde}
\end{align}
Finally, we can re-interpret the degrees of freedom in these equations to live in a Hilbert space of dimension $(L+2)2^{m}$ with a clock register and an $m$-qubit working register, in which case the initial condition
$\big\langle \dot q_j^{(l)}(0) \big\rangle =  \delta_{l,1} \delta_{j,0}$
can be thought of as a quantum state $\ket{l=1}\otimes \ket{0_P0_Q0_R}$, which evolves (non-unitarily) under the equations \eqref{eq:timeevol-atilde}.

\subsection{Separation of Quantum Circuit Logic and Dynamics on the Clock Register}
\label{sec:tight_binding}

To analyze the matrix $\tilde A$ which generates the dynamics, we
separate the quantum circuit logic,
encoded in the $U_l$, from the hopping Hamiltonian acting on the clock register
$\{\ket{l}\}_{l=1}^{L+2}$. This can be achieved through  the basis
transformation
\begin{align}
\label{eq:Celf}
    S &= \sum_{l=1}^{L+2} \ket{l}\bra{l}  \otimes U_{L+1} \cdots U_{l}\ .
\end{align}
The transformed initial state reads $S (\ket{l=1} \otimes \ket{0_P 0_Q0_R}) =
\ket{l=1} \otimes (U \ket{0_P0_Q 0_R})$, with  $U= \prod_{l=1}^{L+1} U_l$. It thus directly
encodes the output state $U \ket{0_P0_Q0_R}$ of the computation. Note that
\begin{align}
    U \ket{0_P0_Q0_R} =: \sqrt{1 - \delta^2} \ket{0_P \Psi_0} + \delta \ket{1_P \Psi_1}\ ,
    \label{eq:delta}
\end{align}
where $\ket{\Psi_1}:=\tfrac{1}{\delta}\langle 1_P\rvert U\ket{0_P 0_Q 0_R}$ 
is the normalized postselected state which encodes the final result
$\|\langle1_Q\vert\Psi_1\rangle\|^2$ of the \textsc{postselected circuit} problem, with postselection probability $\delta^2=\Omega(2^{-L})$.

In the new basis, the matrix $\tilde A$ generating the dynamics becomes $X=S\tilde A S^\dagger$, where
$X$ is of the form
\vspace{.5ex}
\begin{equation}
\label{eq:postsel-initialstate-Stransformed}
    X = X_0  \otimes \ket{0_P}\bra{0_P} 
    \otimes \mathds{1}_{m-1}
    + X_1  \otimes \ket{1_P}\bra{1_P}
    \otimes \mathds{1}_{m-1}\ ,
    \mbox{\ with\ } X_i = \begin{pmatrix} 
        4 & -1 & \dots & 0 & 0 & 0 \\
        -1 & 4 & \dots & 0 & 0 & 0 \\
        \vdots &\vdots &\ddots&\vdots&\vdots & \vdots \\
        0 & 0 & \dots & 4 & -1 & 0 \\
        0 & 0 & \dots & -1 & 4 & -1  \\
        0 & 0 & \dots & 0 & -1 & 4 - \beta^2 \delta_{i,1}
    \end{pmatrix}.
\end{equation}
\vspace{.5ex}
Here, $\mathds{1}_{m-1}$ denotes the identity operator on the remaining $m-1$
qubits (registers $Q$ and $R$). Importantly, we find that the entire dependence
on the circuit $U$ has been absorbed into the basis change $S$, and the
correspondingly updated initial state
\eqref{eq:postsel-initialstate-Stransformed}.

In the new basis, $X$ describes a dynamics which acts purely on the clock
register, and which is controlled solely by the value of the postselected qubit
$\{\ket{0_P},\ket{1_P}\}$. The analysis of the dynamics thus  reduces entirely
to studying the properties of the clock operators $X_0$ and $X_1$. For the sake
of the proof, we can absorb the transformation $S$ in a redefinition of
coordinates $\dot q \to \dot q' = S \dot q$, and we will henceforth work in this
transformed basis. To understand the dynamics of the system, $\dot q'(t) =
\cos(\sqrt{X}t) \dot q'(0)$, in the new basis, we now only need to analyze the
properties of $X$.

Importantly, note that due to the way in which $S$ is constructed, the state of the
computational registers $PQR$ at time $l=L+1$, which is what we choose for the read-out, is the same in the original and the transformed basis, such that we
can analyze the read-out in the new basis.

\subsection{Dynamics of the clock register: Tight binding model}
The matrices $X_0$ and $X_1$ are finite Jacobi matrices. Intuitively, $X_0$
describes a single particle hopping on a one-dimensional chain; this can also be
considered as the lattice discretization of a free particle on a finite line.
$X_1$ differs from $X_0$ in that it additionally has a single-site potential
well at the right end of the chain. If this well is deep enough (i.e., $\beta$
is large enough), it will thus possess an exponentially localized bound state
at the right end of the chain. (Due to the lattice discretization, where the
well is localized at a single site, there can only be one such bound state.)

More precisely, the properties of $X_0$ and $X_1$ and the implications on the
resulting dynamics are as follows (see Section~\ref{sec:walker}
for the detailed derivation): The spectrum of $X_0$ lies in the interval
$[2;6]$; the resulting time evolution, Eq.~\eqref{eq:timeevol-atilde} is therefore unitary, and the operator $\cos(\sqrt{X_0} t)$ is bounded by $\norm{\cos(\sqrt{X_0} t)} \leq 1$. $X_1$ differs from $X_0$ precisely in the rank-one boundary
perturbation $-\beta^2\delta_{i,1}$ in the bottom right corner. It is well known
that rank-one boundary perturbations to lattice Jacobians can produce at most
one eigenvalue outside the spectral band of the unperturbed operator, and thus
at most one bound state (i.e., a negative eigenvalue), see for example
\cite{Simon1995RankOne} or \cite{Teschl2000Jacobi}.

The scenario with such a negative eigenvalue is precisely the case we are
interested in: If $X_1$ has such a  negative eigenvalue, this will give rise to
an exponentially growing eigenstate projector in the time evolution operator
$\cos(\sqrt{X_1} t) = \cosh(\sqrt{-X_1} t)$, which will therefore dominate the
full evolution $\cos(\sqrt{X}t)=\cos(\sqrt{X_0}t)\oplus \cos(\sqrt{X_1}t)$ at sufficiently long times, and thus amplify
the postselected sector. As we show in Subsection~\ref{sec:walker}, the
corresponding eigenvector is indeed exponentially localized at the right
boundary,  yet still has an at least exponentially small weight also at $l=1$.
As we analyze in the next subsection, this ensures that the initial state
\eqref{eq:Celf}, for which $\delta^2=\Omega(2^{-L})$, has sufficient overlap
with the corresponding eigenvector to make the postselected sector dominant
after propagation by a time linear in $L$, and to have sufficient weight on the
time $l=L+1$. 

While the Jacobi matrices $X_0$ and $X_1$ describe standard
tight-binding models, which have well-known solutions, we show these required
properties explicitly in Section \ref{sec:walker} for the sake of completeness. 

\subsection{Time Evolution and Read-Out}

After evolution for a time $t$, the state $\big\langle \dot q(t)\big\rangle = \cos(\sqrt{X}\,t)
\big\langle\dot q(0)\big\rangle$, expressed in the quantum register notation, will be 
\begin{align}
\ket{\Phi(t)} &= \cos(\sqrt{X}\,t)\,\Big[\ket{l=1}\otimes U\ket{0_P0_Q0_R}\Big]
\\
& \stackrel{\eqref{eq:delta}}{=}
\cos(\sqrt{X}\,t)\,\Big[\ket{l=1}\otimes 
    (\sqrt{1-\delta^2}\ket{0_P\Psi_0}+\delta\ket{1_P\Psi_1})\Big]
\\
& =
\sqrt{1-\delta^2}\,\Big[\cos(\sqrt{X_0}\,t)\,\ket{l=1}\Big]\otimes  \ket{0_P\Psi_0} 
    +
\delta\, \Big[\cos(\sqrt{X_1}\,t)\,\ket{l=1}\Big]\otimes  \ket{1_P\Psi_1} \ .
\end{align}
We now use the following: First, $X_0$ is positive semi-definite, and therefore $\cos(\sqrt{X_0}\,t)$ is bounded. Second, we choose
$\beta$ such that $X_1$ has precisely one negative eigenvalue $\alpha_1<0$, with
corresponding eigenvector $\ket{\chi_1}$, see 
Lemma~\ref{lem:Xspectrum} in Section~\ref{sec:walker}; 
thus, $\cos(\sqrt{X_1}\,t) = \cosh(\sqrt{-\alpha_1}\,t)
+ V(t)$, where $V(t)$ is once again norm-bounded by one. Together, this yields the estimate 
\begin{align}
\ket{\Phi(t)} &= 
\delta\,\cosh(\sqrt{-\alpha_1}\,t)\,\ket{\chi_1}\langle\chi_1\vert l=1\rangle\otimes  \ket{1_P\Psi_1}  + O(1)\ .
\end{align}

Let us now define
\begin{align}
\ket{\Phi'(t)} :=&\:
\big(\bra{l=L+1}\otimes\bra{1_P}\big)\,\ket{\Phi(t)} 
\\
=  &\:
\underbrace{\delta\,\cosh(\sqrt{-\alpha_1}\,t)\,\langle l=L+1\ket{\chi_1}\langle\chi_1\vert l=1\rangle}_{=:W(t)}\otimes  \ket{\Psi_1}  + O(1)\ ,
\label{eq:postsel-defPhiPrime-2}
\end{align}
the state obtained after projecting $\ket{\Phi(t)}$ onto time step $l=L+1$ and
the state $\ket{1_P}$ to be postselected on. 
In this notation,  the read-out \eqref{eq:postbqp-readout-indexsets} then
amounts to measuring $\zeta(I_\mathrm{den};t) = \big\|\ket{\Phi'(t)}\big\|^2$
and $\zeta(I_\mathrm{num};t) = \big\|\langle 1_Q\ket{\Phi'(t)}\big\|^2$,
respectively. We required that our measurements gave estimates with $O(1)$
accuracy; this finite precision can be readily absorbed in the $O(1)$ term in
\eqref{eq:postsel-defPhiPrime-2} and thus does not need to be accounted for separately.

Before putting things together, we still need to bound the prefactor $W(t)$ introduced in Eq.~\eqref{eq:postsel-defPhiPrime-2} from below. 
In Lemma~\ref{lem:Xspectrum} in Section~\ref{sec:walker}, we show that for
$\beta>2$, we have $-\alpha_1=2\cosh(\kappa_1)-4\ge \beta^2-4>0$, 
$|\langle l=1|\chi_1\rangle|^2 = \Omega\big(\tfrac{1}{L+2}e^{-2L\kappa_1}\big)$, and 
$|\langle l=L+1|\chi_1\rangle|^2 = \Omega(1)$. Combining this with $\delta^2=\Omega(2^{-L})$ and $\cosh(x)^2=\Omega(e^{2x})$, we obtain
\begin{align}
    |W(t)|^2 = \Omega\Big(\tfrac{1}{L+2}\,e^{2\sqrt{-\alpha_1}\,t-(2\kappa_1+\log 2)\,L}\Big)
\end{align}
and thus, for any $\beta>2$, there are constants $c>0$, $c'>0$ such that for
$t_f\ge cL=O(L)$, it holds that $W(t_f) = \Omega(e^{c'L})$. We are now ready to complete the reduction. Write $W(t_f) = \delta \cdot g(L,\beta,t_f)$, where
\begin{align}
    g(L,\beta,t_f) := \cosh(\sqrt{-\alpha_1}\,t_f)\,\langle l{=}L{+}1|\chi_1\rangle\langle\chi_1|l{=}1\rangle
\end{align}
depends only on the known parameters $L,\beta,t_f$ --- via the spectral data of the clock operators $X_0, X_1$ --- and can be computed classically to inverse-exponential precision in $\mathrm{poly}(L)$ time by numerically solving the quantization condition Eq.~\eqref{eq:beta_condition}. The remaining factor $\delta$ is the postselection amplitude of the circuit $U$ itself; it is a genuine unknown tied to the problem instance and cannot be precomputed. It therefore does not cancel from either population individually, but only from their \emph{ratio}:
\begin{align}
\frac{\zeta(I_\mathrm{num};t_f)}{\zeta(I_\mathrm{den};t_f)} = 
\frac{\big\|W(t_f)\langle 1_Q\vert \Psi_1\rangle + O(1)\big\|^2}{
\big\|W(t_f)\,\lvert\Psi_1\rangle+O(1)\big\|^2}
=
\|\langle 1_Q\vert \Psi_1\rangle\|^2 + O\big(\tfrac{1}{|W(t_f)|}\big)\ ,
\label{eq:postsel-ratio}
\end{align}
using $\|\,\lvert\Psi_1\rangle\|=1$. This is why the reduction estimates two populations rather than one, and correspondingly takes the form of a Turing reduction to Problem \ref{prob:QO} rather than a many-one reduction.

Concretely, both $\zeta(I_\mathrm{num};t_f)$ and $\zeta(I_\mathrm{den};t_f)$ lie in the interval $[\lambda_{\min},\lambda_{\max}]$ with $\lambda_{\min} = \Omega(2^{-L})\,|g(L,\beta,t_f)|^2$ and $\lambda_{\max} = O(1)\,|g(L,\beta,t_f)|^2$, so that $\log_2(\lambda_{\max}/\lambda_{\min}) = O(L) = \mathrm{poly}(m)$; both endpoints are computable classically once $g(L,\beta,t_f)$ is known. Setting $u := \Theta(1/\polylog(M))$ to match the minimal gap required by Problem~\ref{prob:QO}, we query Problem~1 at the sequence of thresholds $a_j = 2^{j+1}u$, $b_j = 2^{j}u$ for $j=0,1,\dots,J$, with $J=O(L)=\mathrm{poly}(m)$ chosen large enough that $2^J u \ge \lambda_{\max}$. Each such query is a valid instance of Problem~1, since $a_j - b_j = 2^j u \ge u = \Omega(1/\polylog(M))$. A doubling search over $j$ using these $O(L)$ queries locates each population to within a constant multiplicative factor. Given such estimates $\hat\zeta_\mathrm{num}$ and $\hat\zeta_\mathrm{den}$, and using that the correction term in Eq.~\eqref{eq:postsel-ratio} is $O(e^{-cL})$ for $t_f\ge cL$ as shown above, the ratio $\hat\zeta_\mathrm{num}/\hat\zeta_\mathrm{den}$ determines $\|\langle 1_Q|\Psi_1\rangle\|^2$ to within an error that can be made smaller than the constant promise gap $1/3$ separating YES and NO instances of \textsc{postselected circuit}, by choosing $t_f=cL$ with $c$ sufficiently large. Comparing $\hat\zeta_\mathrm{num}/\hat\zeta_\mathrm{den}$ to $1/2$ therefore correctly decides the \textsc{postselected circuit} instance using $O(\mathrm{poly}(m))$ queries to Problem \ref{prob:QO} and polynomial-time classical postprocessing. This establishes a polynomial-time Turing reduction from \textsc{postselected circuit} to Problem \ref{prob:QO}, completing the proof.

\subsection{Solution of the Tight-Binding Model on the Clock Register}
\label{sec:walker}
The purpose of this Section is to establish the technical bound underlying the amplitude separation induced by the final-layer perturbation. We solve the spectrum of the matrices $X_0$ and $X_1$, defined in Section \ref{sec:tight_binding}. The spectral properties of the matrices $X_0$ and $X_1$ follow from standard results for finite Jacobi matrices with boundary perturbations. We include the derivation here for completeness and to make all constants explicit.

The two blocks $X_0$ and $X_1$, which act on the subspace of the post-selected qubit in the state $\ket{0_P}$ or $\ket{1_P}$ respectively, can be solved separately. The eigenvectors and -values of $X_0$ are the same as in \cite{Babbush_2023}, i.e. 
\begin{align}
    &X_0 \ket{\phi_l} = \gamma _l \ket{\phi_l}, \nonumber \\
    &\text{with} \quad \ket{\phi_l} = \sqrt{\frac{2}{L+3}} \sum_{l'=1}^{L+2} \sin\left( \frac{ll' \pi}{L+3} \right) \ket{l'} \quad \text{and} \quad \gamma_l = 4 - 2 \cos\left( \frac{\pi l}{L+3} \right).
\end{align}
The operator $X_1$ acts on the subspace of $\ket{1_P}$ similar to $X_0$, but with an additional term $-\beta^2 \ket{L+2}\bra{L+2}$. The spectrum of $X_1$ can be solved analogously to the $\beta = 0$ case using standard methods. The eigenvectors $\ket{\psi_{l}} = \sum_{l'=1}^{L+2} \psi_{ll'} \ket{l'}$ satisfy the difference equations
\begin{align}
    4 \psi_{ll'} - \psi_{l(l'+1)} - \psi_{l(l'-1)} &= \alpha_l \psi_{ll'} \quad \forall l' \in \{1, ..., L+1 \} 
    \label{eq:difference1} \\
    (4 - \beta^2) \psi_{l(L+2)} - \psi_{l(L+1)} &= \alpha_l \psi_{l(L+2)}
    \label{eq:difference2}
\end{align}
using $\psi_{l0} = 0$ for compact notation. The difference equations are solved by the following ansatz $\psi_{ll'} = c_l \sin(l' \omega_l)$ and $\alpha_l = 4 - 2 \cos(\omega_l)$. This can be quickly checked by inserting into Eq.~\eqref{eq:difference1} (omitting the normalization constant $C$)
\begin{align}
    &4 \sin(l' \omega_l) - \sin((l'+1) \omega_l) - \sin((l'-1) \omega_l) \nonumber \\
    &= 4 \sin(l' \omega_l) - \sin(l' \omega_l) \cos(\omega_l) - \cos(l' \omega_l) \sin(\omega_l) - \sin(l' \omega_l) \cos(\omega_l) + \cos(l' \omega_l) \sin(\omega_l) \nonumber \\
    &= (4 - 2 \cos(\omega_l)) \sin(l' \omega_l),
\end{align}
using the angle addition theorem, $\sin(x \pm y) = \sin x \cos y \pm \cos x \sin y$ for all $x,y \in \mathds{C}$. Inserting the ansatz into Eq.~\eqref{eq:difference2} yields a quantization condition on the frequencies $\omega_l$
\begin{align}
    (2 \cos(\omega_l) - \beta^2) \sin((L+2) \omega_l) &= \sin((L+2 - 1) \omega_l) = \sin((L+2)\omega_l) \cos(\omega_l) - \cos((L+2)\omega_l) \sin(\omega_l) \nonumber \\
    & \iff \beta^2 = \cos(\omega_l) + \frac{\sin(\omega_l)}{\tan((L+2) \omega_l)}.
    \label{eq:beta_condition}
\end{align}
In the $\beta = 0$ case, this yields $\tan(\omega_l) = \tan((L+2)\omega_l)$, which determines the frequencies used above. For nonzero $\beta$, there are no analytical solutions to this equation, in general. However, we will show that, for large enough $\beta$, there exists a unique (pair of) solution(s) of $\omega_1 = \pm i \kappa_1$ to Eq.~\eqref{eq:beta_condition} that is purely imaginary and ensures a negative eigenvalue $\alpha_1 = 4 - 2\cosh(\kappa_1) < 0$, and ultimately creates exponential growth in $\cos(\sqrt{X_1} t)$ as opposed to oscillatory dynamics from positive $\alpha_l, \gamma_l$.
\begin{lemma}
\label{lem:Xspectrum}    
Let $\kappa_1 = i\omega_1 \in \mathds{R}$ be a solution of Eq.~\eqref{eq:beta_condition} for $\beta > 2$ and $L \geq 0$.
    \begin{enumerate}
        \item There exists a unique negative eigenvalue $\alpha_1 = 4 - 2\cosh(\kappa_1)$, separated from zero by a constant, $\alpha_1 \leq 4 - \beta^2 < 0$.
        \item The corresponding eigenvector $\ket{\psi_1}$ has an overlap $\abs{\braket{l=1 | \psi_1}}^2 = \Omega\left( \frac{e^{-2L\kappa_1}}{L+2} \right)$.
        \item The corresponding eigenvector $\ket{\psi_1}$ has an overlap $\abs{\braket{l=L+1 | \psi_1}}^2 = \Omega(1)$ that is separated from zero by a constant.
    \end{enumerate}
    \begin{proof}
    \begin{enumerate}
        \item We start from Eq.~\eqref{eq:beta_condition} and bound
        \begin{align}
            \cosh(\kappa_1) = \beta^2 - \sinh(\kappa_1) \coth((L+2)\kappa_1) = \beta^2 - \tanh(\kappa_1) \cosh(\kappa_1) \coth((L+2)\kappa_1),
        \end{align}
        substituting $\sinh(\kappa_1) = \tanh(\kappa_1) \cosh(\kappa_1)$. Now, we solve again for $\cosh(\kappa_1)$ to get
        \begin{align}
            \cosh(\kappa_1) = \frac{\beta^2}{ 1 + \tanh(\kappa_1) \coth((L+2)\kappa_1)} \geq \frac{\beta^2}{2},
        \end{align}
        where we used $\tanh(x) \coth((L+2)x) \leq 1$ for $L\geq 0$. Uniqueness of $\alpha_1$ is a consequence of the reality of the eigenvalues of $X_1$; $\alpha = 4 - 2 \cos(\omega)$ is real if, and only if $\omega$ is purely real or purely imaginary, otherwise $\cos(x + iy) = \cos(x) \cosh(y) - i \sin(x) \sinh(y)$ is not real. For $\omega$ real, $\alpha$ cannot be negative, hence the only possible solution is $\omega = \pm i \kappa$, for which there is only one solution for $\kappa$ through Eq.~\eqref{eq:beta_condition} given a value for $\beta^2$, which is a monotonic function in $\kappa$. Hence exactly one exponentially growing mode exists in the $\ket{1_P}$ sector, while all other modes remain oscillatory.

        \item Consider the normalized eigenvector corresponding to the eigenvalue $\alpha_1$
        \begin{align}
            \ket{\psi_1} = c_1 \sum_{l' = 1}^{L+2} \sinh(l' \kappa_1) \ket{l'} \qquad \text{with} \qquad c_1 = \left(\sum_{l' = 1}^{L+2} \sinh(l' \kappa_1)^2\right)^{-\frac{1}{2}} = 2 \left( \frac{\sinh((2L+5) \kappa_1)}{\sinh(\kappa_1)} - (2L+5) \right)^{-\frac{1}{2}},
        \end{align}
        We bound the overlap with the $l=1$ using the inequality $c_1 \geq \frac{1}{\sqrt{L+2}} \cdot \frac{1}{\sinh((L+2)\kappa_1)}$. Thus, 
        \begin{align}
            \abs{\braket{l=1 | \psi_1}} &\geq \frac{1}{\sqrt{L+2}} \frac{\sinh(\kappa_1)}{\sinh((L+2) \kappa_1)} = \Theta\left( \frac{e^{-L\kappa_1}}{\sqrt{L+2}} \right).
        \end{align}
        \item The third statement can be similarly bounded. Using partial sums of the geometric series, one can calculate $c_1 = 2 \left( \frac{\sinh((2L+5) \kappa_1)}{\sinh(\kappa_1)} - (2L+5) \right)^{-\frac{1}{2}}$. With this, we can bound
        \begin{align}
            \abs{\braket{l=L+1 | \psi_1}}^2 = 2 \frac{\sinh((L+1)\kappa_1)^2}{ \frac{\sinh((2L+5) \kappa_1)}{\sinh(\kappa_1)} - (2L+5) } \geq 2 \sinh(\kappa_1) \frac{e^{2(L+1)\kappa_1} - 2}{ e^{(2L+5) \kappa_1} } \geq 2 \sinh(\kappa_1) e^{-3\kappa_1} (1 - 2 e^{-2 \kappa_1} ) ,
        \end{align}
        by dropping positive terms in the numerator and negative terms in the denominator. This lower bound has a constant separation from zero.
    \end{enumerate}
    \end{proof}
\end{lemma}
The crude bound $\tanh(x) \coth((L+2)x) \leq 1$ in the proof of statement one of Lemma \ref{lem:Xspectrum} leads us to the constraint $\beta > 2$, but can be tightened to the $L$-dependent bound
\begin{align}
    \beta^2 > 2 + \sqrt{3} \; \frac{(2 + \sqrt{3})^{2(L+2)} + 1}{(2 + \sqrt{3})^{2(L+2)} - 1},
    \label{eq:beta_bound}
\end{align}
using the critical value $\kappa_0 = \cosh^{-1}(2) = \log(2 + \sqrt{3})$. Eq.~\eqref{eq:beta_bound} includes the value $\beta = 2$, for instance, to ensure a constant separation of $\alpha_1$ from zero. Using $\kappa_1 \geq \kappa_0$, also yields a numerical lower bound of $5\%$ for the third statement of Lemma \ref{lem:Xspectrum}.

\subsection{Alternative PostBQP-hard families of bosonic Hamiltonians}
\label{app:alt_PostBQP}
At the end of this appendix, let us argue that the hardness result is robust under more general choices of $E^+$. One might wonder how restrictive the choice of nonzero $E^+$ is for \PostBQP-hardness. Indeed, \PostBQP-hard instances can be found in a more general class of Hamiltonians. Off-diagonal terms $q_j p_k + q_k p_j$ of $E^+$ can be diagonalized by changing the qubit encoding using a symplectomorphism $q \mapsto O q$ and $p \mapsto (O^{-1})^T p$. If we were to only switch on a single term, post-selection would not target a single-qubit output $\ket{1_P}$ but rather a whole bitstring, which is equivalent to post-selecting only one qubit at the end and can be mapped onto another using an ancillary register \cite{aaronson2004quantumcomputingpostselectionprobabilistic}. 

Another family of Hamiltonians whose simulation is \PostBQP-hard is given by the constraint $E^+ = 0$. As we show in appendix \ref{app:coordinates}, besides $E^+$ also $\Delta = A - C$ can generate non-unitary dynamics. Setting $E = i E^-$, the dynamics is governed by the differential equation
\begin{align}
    \ddot q = -i E^- \dot q + C \dot p = -\left( (E^-)^2 + CA \right) q - i (E^- C + C E^-) p.
\end{align}
If we can get rid of the term proportional to $p$, we can reproduce Eq.~\eqref{eq:diff_HO}, in particular the differential operator $\tilde A$. In particular, we have to choose $A, C$ real symmetric and $E^-$ real antisymmetric, such that 
\begin{enumerate}
    \item $E^- C + C E^- = 0$
    \item $CA$ acts as a Feynman-Kitaev Hamiltonian
    \item $(E^-)^2 = -\beta^2 \Pi$ with an adequate projector $\Pi$ onto the postselected part of the readout layer
\end{enumerate}
These conditions can be implemented, for instance, by defining two Feynman-Kitaev dynamics on a doubled system
\begin{align}
    A &= A_{FK} \otimes Z \\
    C &= \mathds{1} \otimes Z \\
    E^- &= \beta \ket{L+2} \bra{L+2} \otimes \ket{1_P} \bra{1_P} \otimes (iY)
\end{align}
with $A_{FK}$ being the matrix defined in Eq.~\eqref{eq:A_FK} and $Z$ and $Y$ are Pauli matrices. One can directly check that (1.) $E^-$ and $C$ anticommute, (2.) $CA$ describes two copies of the same Feynman-Kitaev Hamiltonian and (3.) $(E^-)^2 = - \beta^2 \Pi_{L+2, 1_P} \otimes \mathds{1}$ postselects onto the last readout layer in both copies.
\begin{align}
    (E^-)^2 + CA = \tilde A \otimes \mathds{1},
\end{align}
reproduces the generator used in the proof above. It has the same spectral properties -- as before -- and when starting in the analog initial state with only $\braket{p^{(l=1)}_{0}(0)} = 1$, one can solve the \textsc{postselected circuit} problem by estimation of the momenta in the output layer. Note that while $A$ and $C$ individually are indefinite, the product $CA$ only has positive eigenvalues, which shows that the non-unitary dynamics necessary for postselection comes solely from $(E^-)^2$, cf. also appendix \ref{app:PSD}. 

A third family of Hamiltonians which are \PostBQP-hard is given by $E=0$, but lifting the condition that $CA$ has to have a positive spectrum, as we require in systems of inertially coupled bosons. In this family of Hamiltonians one can directly choose $A$ to be the targeted generator $\tilde A$ as in Eq.~\eqref{eq:diff_HO} with a negative eigenvalue, which is now allowed.

\section{Coordinate Transformations}
\label{app:coordinates}
This appendix analyzes the conditions under which the first-moment dynamics of quadratic bosonic systems remain bounded, i.e., do not exhibit exponential
amplification. While the main text focuses on inertially coupled bosons, the
following discussion treats general quadratic Hamiltonians and later specializes
to the families shown in Fig.~\ref{fig:QO} of the main text. Finally, we demonstrate how
coordinate transformations can be used to interchange terms in the resulting
differential equations, thereby identifying the boundaries between simulable
families of Hamiltonians.

\subsection{Most General Quadratic Bosonic Hamiltonian}
The most general quadratic Hamiltonian can be written as
\begin{align}
    H = \frac{1}{2} \begin{pmatrix} q & p \end{pmatrix} \begin{pmatrix} A & E \\ E^T & C \end{pmatrix} \begin{pmatrix} q \\ p \end{pmatrix},
    \label{eq:quadHam}
\end{align}
where the matrix $E=E^+ + i E^-$ can be split into a real and symmetric part $E^+$ contributing with terms $\frac{1}{2} E_{jk} (q_j p_k + q_k p_j)$ and an imaginary, antisymmetric (i.e.\ Hermitian) part $E^-$ contributing with terms of the form $\frac{1}{2} E_{jk} (q_j p_k - q_k p_j)$. In matrix form, the equations of motion become
\begin{align}
    \partial_t \begin{pmatrix} i q \\ p \end{pmatrix} = -i \begin{pmatrix} i E^T & -C \\ -A & -i E \end{pmatrix} \begin{pmatrix} i q \\ p \end{pmatrix} = -i \left( \mathds{1} \otimes E^- + (iZ) \otimes E^+ + (iY) \otimes \Delta - X \otimes \Sigma \right) \begin{pmatrix} iq \\ p \end{pmatrix},
    \label{eq:qp_dynamics}
\end{align}
with $\Sigma= \frac{1}{2} (A + C)$ and $\Delta = \frac{1}{2} (A - C)$. For the first-moment dynamics to remain bounded—i.e., to avoid exponential amplification or suppression of the variables of interest—the matrix in Eq.~\eqref{eq:qp_dynamics} must have purely real eigenvalues, corresponding to oscillatory eigenmodes. A sufficient condition ensuring this is that the generator in Eq.~\eqref{eq:qp_dynamics} is anti-Hermitian.

\subsection{Positivity of coupling matrices $A$ and $C$}
\label{app:PSD}
As we show in the main text, coordinates can be chosen as $B^\dagger q$ and $D^\dagger p$ for $A= BB^\dagger$ and $C=DD^\dagger$, such that correlations evolve unitarily in the absence of mixed terms, $E=0$. If in addition to positivity of $A$ and $C$, the matrix elements are derived from positive coupling constants (cf. Eqs.~\eqref{eq:sign_convention_I} and \eqref{eq:sign_convention_II} of the main text), an explicit and efficiently queriable form of $B$ and $D$ is known. If this sign convention is not met, an oracle access to quantum circuits implementing $\sqrt{A}$ and $\sqrt{C}$ can be employed to simulate time evolution via quantum phase estimation, cf. Theorem~4 of \cite{Babbush_2023}.

In fact, the condition that both, $A$ and $C$, are PSD is stronger than necessary. Assuming $A,C$ invertible (a nontrivial shared kernel generically leads to secular, linearly growing modes and is excluded here), the correct condition for the first-moment dynamics to remain bounded is that $AC$ (equivalently $CA$, to which it is always similar via conjugation by $C$) is diagonalizable with real, non-negative spectrum; a positive spectrum of $AC$ alone does not suffice, since diagonalizability is not automatic. Indeed, $\mathcal{L}^2 = -(CA\oplus AC)$ for the generator $\mathcal L$ of Eq.~\eqref{eq:qp_dynamics} restricted to $E=0$, so nonzero eigenvalues of $AC$ lift directly to eigenvalues $\lambda=\pm i\sqrt\mu$ of $\mathcal L$; this lift preserves diagonalizability precisely when $AC$ itself is diagonalizable, and yields purely imaginary (oscillatory) $\lambda$ precisely when the corresponding eigenvalue $\mu$ of $AC$ is real and non-negative.

Diagonalizability of $AC$ is guaranteed under either of the following sufficient conditions. First, if at least one of $A,C$ is sign-definite --- say $C\succ0$ --- then $AC$ is similar to the symmetric matrix $C^{1/2}AC^{1/2}$ via conjugation by $C^{-1/2}$, and is therefore diagonalizable with real eigenvalues for \emph{arbitrary} symmetric $A$; this recovers the case of two positive definite matrices as a special case, and the eigenvalues are non-negative additionally whenever $A\succeq0$. Second, if $A,C$ are both indefinite but share a common invariant orthogonal decomposition $\mathbb{R}^n=V_+\oplus V_-$ on which their signs are matched (both positive definite on $V_+$, both negative definite on $V_-$), the same similarity argument applies block-wise, since a direct sum of symmetric matrices remains symmetric.

Beyond these two cases, diagonalizability of $AC$ is not automatic: for generic indefinite $A,C$ without either sign-definiteness or a compatible shared block structure, $AC$ can develop genuine Jordan blocks, at which point the dynamics acquires a polynomially growing, rather than purely oscillatory or exponential, component. An example with indefinite spectrum is the single-mode squeezing Hamiltonian 
\begin{align}
    H = \frac{1}{2} (q^2 - p^2) = \frac{1}{2} (a^2 + (a^\dagger)^2),
\end{align}
corresponding to $A=1$ and $C=-1$. In this case the dynamics exponentially amplifies one quadrature, while suppressing the conjugate one. Effectively, this projects our encoded state $\ket{\psi}$ onto the squeezed quadrature eigenstate, $\frac{1}{\sqrt{2}} \left( \ket{\alpha = 1} + \ket{\alpha = \bar 1} \right)$.

\subsection{Symplectic Transformations and Generalized Hopping}
Upon inspection of Eq.~\eqref{eq:qp_dynamics}, $R$-point correlations of $\braket{iq}$ and $\braket{p}$ can also be directly simulated (without transforming into the coordinates $B^\dagger q$ and $D^\dagger p$) by switching off the non-Hermitian terms $E^+ = \Delta = 0$. This extends the family of quantum walk Hamiltonians discussed in the main text (cf. Eq.~\eqref{eq:quantum_walk}) by extending the graph to complex weights, where imaginary parts correspond to $E^- \neq 0$. \BQP-containment for the class $\Delta = E^+ = 0$ can then be proven analogously to Theorem \ref{thm:BQP} using encoding and read-out in the coordinates $(iq, p)$ and the sparsity of the Hamiltonian from Eq.~\eqref{eq:qp_dynamics}.

Using unitary transformations on the Paulis in Eq.~\eqref{eq:qp_dynamics} allows for a permutation of terms and permutations of the phase factor $i$, which determines Hermiticity. In particular, a symplectic transformation, that is a coordinate transformation which leaves the symplectic form $\Omega = (iY) \otimes \mathds{1}$ unchanged, also leaves the equations of motion invariant. One example is the transformation 
\begin{align}
	-i \exp\left(-i \frac{\pi}{4} Y \right) \otimes \mathds{1} = \frac{-i}{\sqrt{2}} \begin{pmatrix} \mathds{1} & -\mathds{1} \\ \mathds{1} & \mathds{1} \end{pmatrix},
\end{align}
which maps the $(iq,p)$ basis onto the Fock basis spanned by ladder operators $(a,a^\dagger)$. These coordinates were used in the main text (cf. discussion around Eq.~\eqref{eq:quantum_walk_pre}) to show the equivalence of this simulable family, $E^+ = \Delta = 0$ to number-preserving Hamiltonians. 

\section{Simulation of Quadratic Bosonic Hamiltonians is in \PP}
\label{app:PP}

In this appendix, we prove that the simulation of exponentially large quadratic bosonic Hamiltonians is in \PP. Concretely, we show that the following promise problem lies in \PP. Given a $d$-sparse matrix $A \in \mathrm{Mat}(M \times M, \mathds{C})$ with elements $\abs{A_{jk}} = \polylog(M)$, an $M$-dimensional vector $x(0)$, both with efficiently computable entries, and an efficiently describable index set $I$, define 
\begin{align}
    \zeta(I,t)=\sum_{i\in I} \left|(e^{At}x(0))_i \right|^2
    \label{eq:PPzeta}
\end{align}
For $t = \polylog(M)$, decide whether $\zeta(I,t) > a$ or $\zeta(I,t) < b$, promised that one holds and $a - b = 1 / \polylog(M)$. Any version of Problem \ref{prob:QO}, with any Hamiltonian family, is contained in this problem. To solve this problem, we express $\zeta(I,t)$ as a truncated sum over pairs of paths of length at most $\polylog(M)$, where each term is efficiently computable. After scaling to integers, this sum can be encoded as a \GapP \ function, whose sign can be decided in \PP.

\subsection{Series Truncation}
First, let us write $\zeta(I,t)$ as a quadratic form. Let $\Pi_I = \sum_{i \in I} \ket{i} \bra{i}$ be the projector onto indices in $I$. Then
\begin{align}
    \zeta(I,t)=x(0)^\dagger e^{A^\dagger t}\,\Pi_I\, e^{At} x(0).
    \label{eq:zeta1}
\end{align}
The exponential function can be approximated by the order $K$ truncation \cite{Higham2008}:
\begin{align}
    e^{At}=\sum_{k=0}^{K}\frac{(At)^k}{k!} + R_K, \qquad \text{with } \epsilon := \norm{R_K} \;\le\; e^{\|A\|t}\frac{(\|A\|t)^{K+1}}{(K+1)!} 
    &= \mathcal{O} \left( \frac{e\|A\|t}{K} \left(\frac{e\|A\|t}{K} e^{\norm{A} t / K} \right)^K \right),
    \label{eq:eAt}
\end{align}
using the Stirling bound. Choosing $K=c(\|A\|t+\log\log M)$ with $c>e^2$ yields $\|A\|t/K=\mathcal O(1/c)$ and thus $\epsilon = \mathcal O((e^2/c)^{\log\log M})=\mathcal{O}(1/\polylog(M))$. For $t=\polylog(M)$ and $\|A\| = \mathcal{O}(d \cdot \max_{ij}|A_{ij}|) = \polylog(M)$, one can take $K=\polylog(M)$ to achieve $\epsilon = \mathcal{O}(e^{-\polylog(M)})$. Define $y_K(t) := \left(\sum_{k=0}^{K}\frac{(At)^k}{k!}\right) x(0)$, the
order-$K$ truncation of $e^{At}x(0)$ prior to any rounding of matrix entries,
so that $e^{At}x(0) = y_K(t) + R_K x(0)$ with $\|R_K x(0)\| \le \epsilon\|x(0)\|$.

Plugging Eq.~\eqref{eq:eAt} into Eq.~\eqref{eq:zeta1} yields
\begin{align}
\zeta(I,t)=
\sum_{i\in I}\sum_{k_1,k_2=0}^{K}
\sum_{j_1,\ldots,j_{k_1+k_2}}
\frac{t^{k_1+k_2}}{k_1!\,k_2!}\;
x(0)_{j_1}\,
A^\dagger_{j_1 j_2}\cdots A^\dagger_{j_{k_1} i}\,
A_{i j_{k_1+1}}\cdots A_{j_{k_1+k_2-1} j_{k_1+k_2}}\,
x(0)_{j_{k_1+k_2}}
\;+\;\mathcal{O}(\epsilon).
\label{eq:zeta}
\end{align}

\subsection{Counting Paths}
To account for finite machine precision, we replace all entries of $t A$ and $t A^\dagger$ in Eq.~\eqref{eq:zeta}  with their nearest fraction with denominator $2^{Q}$ for $Q = \polylog(M)$, introducing a rounding error $\mathcal{O}\left( K 2^{-Q} \norm{tA}^{K-1} \right) = \mathcal{O}\left( e^{-\polylog(M)} \right)$ per term in the sum. The product of at most $K$ such terms has a common denominator $D=2^{KQ} (K!)^2$ 
which is $\mathcal{O}(e^{\polylog(M)})$. In a later step, we need each term in the sum to be an integer, and thus multiply $\zeta(I, t)$ by the common denominator $D$.  

Each term corresponds to a pair of paths $(\vec j_1, \vec j_2)$ of lengths $k_1, k_2 \le K$ (total length $\le 2K$). Each index runs from 1 to $M$ (counting both zero and nonzero matrix elements), so the $\mathcal{O}\left( M^{2K} \right)$ paths can be specified using $\mathcal{O}(K\log M)$ bits.

Given such an encoding, the corresponding contribution (a product of polynomially many terms, that are $\mathcal{O}(K)$ matrix entries, two components of $x(0)$, and $\mathcal{O}(K)$ many rational prefactors) can be computed in time $\polylog(M)$. Thus, each term in the expansion is efficiently computable from its nondeterministic description, and the full sum is represented as an implicit aggregation over all such encoded path pairs. Write the sum as
\begin{align}
    \zeta(I,t)=\sum_{(\vec j_1, \vec j_2)} \omega(\vec j_1, \vec j_2), \qquad \omega(\vec j_1, \vec j_2) \in \mathds{C}.
\end{align}
Since the rounded entries of $tA$ and $tA^\dagger$ are dyadic rationals with denominator $2^Q$, each weight $D\cdot\omega(\vec j_1,\vec j_2)$ is a Gaussian integer, $D\cdot\omega(\vec j_1,\vec j_2) \in \mathds{Z}[i]$. We accordingly define the weight
\begin{align}
    w(\vec j_1,\vec j_2) := \Re\big(D\cdot\omega(\vec j_1,\vec j_2)\big) \in \mathds{Z},
\end{align}
obtained by discarding the imaginary part of this Gaussian integer, which is manifestly a real integer for every path pair, not merely in total after summation. Since $\Re(\cdot)$ is $\mathds{R}$-linear,
\begin{align}
    \sum_{(\vec j_1,\vec j_2)} w(\vec j_1,\vec j_2) = D\cdot \Re\left( \sum_{(\vec j_1,\vec j_2)} \omega(\vec j_1,\vec j_2) \right)\ ,
\end{align}
and since the corresponding quantity prior to rounding is exactly real and non-negative (it equals $\|\Pi_I y_K(t)\|^2$, cf. Eq.~\eqref{eq:zeta}), taking the real part discards nothing beyond the rounding error already bounded below. The integrality of weights $w(\vec j_1, \vec j_2)$ enables an encoding of the sum as a \GapP \ function, i.e. a function counting the difference between the number of accepting and rejecting paths of a non-deterministic polynomial-time Turing machine.

\subsection{\GapP \ Realization}
Define a nondeterministic polynomial-time machine $M$ that:
\begin{enumerate}
    \item guesses a pair of paths $(\vec j_1, \vec j_2)$ of lengths $\le K$ (using $\mathcal{O}(K\log M)$ bits);
    \item computes $w=w(\vec j_1, \vec j_2)$ in polynomial time;
    \item if $w>0$, creates $w$ accepting branches; if $w<0$, creates $|w|$ rejecting branches (implemented with $\mathcal{O}(\log |w|) = \mathcal{O}(\polylog(M))$ extra bits).
\end{enumerate}
Formally, creating a number of $\abs{w}$ accepting branches (or rejecting branches, respectively) corresponds to the insertion of a sum $\sgn(w) \sum_{x=1}^{\max(w(\vec j_1, \vec j_2))} g(x,w)$ with $g(x, w) = 1$ if $x \leq \abs{w}$ and 0 else. Then
\begin{align}
    \#\mathrm{acc}(M) - \#\mathrm{rej}(M) = \sum_{(\vec j_1, \vec j_2)} w(\vec j_1, \vec j_2) = D\cdot \zeta(I,t) \pm\ \mathcal{O}\left( e^{-\polylog(M)} \right),
\end{align}
so $D\cdot \zeta(I,t)$ is a \GapP \ function up to controlled additive errors $\mathcal{O}\left( e^{-\polylog(M)} \right)$ if we choose $K=\Theta(\|A\|t+\log\log M) = \polylog(M)$. Path descriptions use $\mathcal{O}(2K\log M) = \polylog(M)$ nondeterministic bits.

After shifting the function by the decision midpoint $T=(a+b)/2$, we can consider
\begin{align}
    F := D\cdot(\zeta(I,t)-T),
\end{align}
which is also a \GapP \ function with promise $|F| \ge D \cdot (a-b)/2 = \Omega(1/\polylog(M)) \cdot D$. By the standard characterization of \PP \ via \GapP, deciding the sign of $F$ (equivalently, whether $\zeta > a$ or $\zeta < b$) is in \PP.

\end{document}